\documentclass[a4paper, 12 pt, twoside, reqno]{amsart}

\usepackage[english]{babel}
\usepackage[utf8]{inputenc}

\usepackage{setspace}
\usepackage[euler-digits,euler-hat-accent]{eulervm}

\usepackage{times}
\usepackage{amsmath}
\usepackage{amsfonts}
\usepackage{amssymb}
\usepackage{amsmath}
\usepackage{amsthm}
\usepackage{mathtools}
\usepackage{graphicx}
\usepackage{array}
\usepackage{color}
\usepackage{xcolor}
\usepackage{mathrsfs}
\usepackage{hyperref}
\usepackage{eucal}
\usepackage{esint} 
\usepackage{tikz}
\usepackage{tikz-cd}
\usepackage{upgreek}
\usepackage{enumitem}

\allowdisplaybreaks

\theoremstyle{plain}
\newtheorem{theorem}{Theorem}[section]
\newtheorem{corollary}[theorem]{Corollary}
\newtheorem{proposition}[theorem]{Proposition}
\newtheorem{lemma}[theorem]{Lemma}

\theoremstyle{definition}
\newtheorem{definition}[theorem]{Definition}

\theoremstyle{remark}
\newtheorem{remark}[theorem]{Remark} 
\newtheorem{example}[theorem]{Example}

\numberwithin{equation}{section}
\numberwithin{figure}{section}
\numberwithin{table}{section}

\newcommand{\R}{\mathbb{R}}
\newcommand{\N}{\mathbb{N}}
\newcommand{\C}{\mathbb{C}}                           

\newcommand{\Z}{\mathbb{Z}}

\newcommand{\s}[1]{\CMcal{#1}}
                  
\newcommand{\bb}[1]{\mathscr{#1}}
\newcommand{\rr}[1]{\mathfrak{#1}}
\newcommand{\n}[1]{\mathbb{#1}}

\newcommand{\expo}[1]{\,\mathrm{e}^{#1}\,}                 

\newcommand{\dd}{\,\mathrm{d}}
\newcommand{ \ii}{\,\mathrm{i}\,}

\newcommand{\virg}[1]{\lq\lq#1\rq\rq}                \newcommand{\ie}{\textsl{i.\,e.\,}}

\newcommand{\etc}{\textsl{etc}.\,}

\hypersetup{
pdftoolbar=true,        
pdfmenubar=true,        
pdffitwindow=true,     
pdfstartview=true,    
pdftitle={NCG-Landau-II},    
pdfauthor={Giuseppe De Nittis},     
breaklinks=true, 
colorlinks=true,       
linkcolor=purple,         
citecolor=teal, 
urlcolor=blue, 
bookmarksopen=true, 
filecolor=magenta,      
}

\begin{document}

\title{The Lieb-Robinson condition and the  Fr\'echet topology}

\author[S. Bachmann]{Sven Bachmann}

\author[G. De~Nittis]{Giuseppe De Nittis}

\author[J. Gomez]{Juli\'an G\'omez}

\address[G. De~Nittis]{Facultad de Matemáticas \& Instituto de Física,
  Pontificia Universidad Católica de Chile,
  Santiago, Chile.}
\email{gidenittis@uc.cl}

\address[S. Bachmann]{Department of Mathematics,
  The University of British Columbia,
  Vancouver, Canada.}
\email{sbach@math.ubc.ca}

\address[J. G\'omez]{Departamento de Matemáticas, 
	Universidad de los Andes, 
	Bogotá, Colombia \& 
	Facultad de Matemáticas, 
	Pontificia Universidad Católica de Chile,
	Santiago, Chile.}
\email{jd.calderon@uniandes.edu.co}

\vspace{2mm}

\date{\today}

\begin{abstract}
We define various notions of locality for *-automorphisms of the algebra of observables for an infinitely extended quantum spin system and study their relationship. In particular, we show that the ubiquitous characterization which arises from the Lieb-Robinson bound implies but is not equivalent to continuity with respect to the natural Fr\'echet topology of almost local observables, which is a non-commutative analog of the Schwartz space. 

\medskip



\end{abstract}

\maketitle

\tableofcontents

\section{Introduction}\label{sec:Intr0}
The study of quantum lattice systems has revealed profound connections between locality, dynamics, and the algebraic structure of observables. One of the most remarkable manifestations of these connections is found in the so-called Lieb-Robinson bounds. First introduced in the early 1970s \cite{Lieb-Robinson Original}, these bounds establish an effective light cone for the Heisenberg dynamics of quantum spin systems with short-range interactions, drawing an analogy with the causal structure of relativistic field theories. The support of the time evolution of a local observable remains exponentially small outside this effective cone, thus offering a quantitative formulation of {\sl almost-locality} in non-relativistic quantum many-body dynamics.

A Lieb-Robinson bound is associated with the quantum dynamics $\{\tau_t^\Phi:t\in\mathbb{R}\}$ generated by the formal Hamiltonian $\sum_{\Lambda\subset\mathbb{Z}^d}\Phi(\Lambda)$.  While the original bound was for finite range or exponentially decaying interactions, see also~\cite{Nachtergaele-Sims Exponential Clustering}, recent generalizations initiated in~\cite{HastingsPolLRB} have been developed to describe and control the dynamics of systems with long-range interactions, namely where the interactions decay polynomially. In all cases, the Lieb-Robinson bound has the form
$$||[\tau_t^\Phi(A), B]||\; \leq\; ||A|| \, ||B|| \,\mathrm{e}^{v\vert t\vert}\sum_{i \in \Sigma_1} \sum_{j \in \Sigma_2} F(d(i,j))\;,$$
where $A,B$ are observables supported on sets $\Sigma_1,\Sigma_2\subset \Z^d$ and $d(i,j)$ is the distance between the sites $i,j\in\Z^d$. The type of decay of the interaction $\Phi$ determines the decay of the non-increasing function $F:[0, \infty) \rightarrow (0, \infty)$. For times in a compact interval $[0,T]$, the exponential can be treated as a constant, yielding our definition of a Lieb-Robinson-type automorphism, Definition~\ref{defLR-aut}. This formulation is very close to that of~\cite{AdiabaticThm}, see also~\cite{Nachtergaele-Sims-Young-Quasilocality-bounds}. 

In this paper, we are not aiming to provide further refinements of the Lieb-Robinson bound. Rather, we are interested more abstractly in the concept of locality of automorphisms of quasi-local algebra of observables, and in particular algebraic and topological characterizations thereof. While the global $C^*$-algebra of observables is obtained by completing the space of local observables in the norm toplogy, one may also consider a dense subalgebra consisting of observables that decay rapidly in space. This subalgebra, which we denote by $\bb{A}_\infty$, can be endowed with the structure of a Fr\'echet algebra and it is a non-commutative analog of the space of Schwartz functions, see Section~\ref{sub:ALO}. It appeared in various guises before: Almost local observables were first defined in algebraic quantum field theory~\cite{araki1967collision}, adapted to the lattice setting in~\cite{bachmann2016lieb}, and used more systematically in~\cite{ogata2021Cohom} and \cite{Kapustin-Sopenko-Local-Noether}.

This Fr\'echet algebra $\bb{A}_\infty$ and the Fr\'echet topology provide a convenient framework for describing almost local observables and the propagation properties of automorphisms. One may be lead to believe that continuity of an automorphism with respect to this topology is equivalent to the automorphism being of Lieb-Robinson type dynamics. We will show that this is not the case: not all continuous automorphisms with respect to this topology arise from dynamics governed by local interactions. Indeed,

\begin{theorem}
Any Lieb-Robinson-type automorphism is continuous on the algebra of almost-local observables with respect to the Fr\'echet topology. However, Fr\'echet continuity alone does not imply the existence of a Lieb-Robinson bound.
\end{theorem}

Besides Lieb-Robinson-type automorphisms, we shall define two different types of `locality-preserving automorphisms', see Definition~\ref{ALP automorphism firts definition} and~Definition \ref{ALP Definition 2}. These are two characterizations of families of automorphisms that, in principle, are not explicitly defined through interaction potentials, yet describe local propagation up to errors of a specific form. Heuristically, these locality-preserving automorphisms are such that if $A$ is strictly supported in a ball of radius $s$, then $\alpha(A)$ can be approximated by an element supported in a ball of radius $s+r$, up to errors that decay fast in the fattening parameter $r$ (and may grow as a function of the initial surface area $s^{d-1}$). Despite their different formulations, these two notions are equivalent, and in fact also equivalent to the Lieb-Robinson-type condition sketched above.

\begin{theorem}
An automorphism is of Lieb-Robinson type if and only if it is almost locality-preserving.
\end{theorem}

Let us immediately point out that part of this discussion has a review character, since it also recalls and adapts known results, see in particular~\cite{Nachtergaele-Sims-Young-Quasilocality-bounds}, and also~\cite{Kapustin-Sopenko-Lieb} and~\cite{Ranarad-Walter-Witteveen-Converse-Lieb}.

Finally, we will consider the polynomial cases, see Section~\ref{Pol Case sec}. From the Lieb-Robinson perspective, these are automorphisms generated by interactions whose decay is only polynomial, which is often referred to as long-range interactions, see for example~\cite{HastingsPolLRB,Kuwahara,Improved LR polynomial}. Here again, the natural algebraic setting is to consider a subalgebra of observables $\bb{A}_{(k)}$ (see Definition~\ref{Def: Polynomial local observables}) whose local approximations decay at least as a power law of degree $k$ with the distance. We shall define $(k)$-almost locality preserving automorphisms and relate them to Lieb-Robinson type bounds for long-range interactions. Even in this larger class of automorphisms, there is no equivalence between a polynomial Lieb-Robinson-type condition and continuity with respect to the natural topology formalizing decay of order $k$ in space.

\medskip
 
 \noindent
{\bf Acknowledgements.}
GD's research is supported by the grant \emph{Fondecyt Regular - 1190204}. SB acknowledges support from NSERC of Canada.


\section{Algebraic description of a spin lattice system}
A \emph{spin lattice system} is a $d$-dimensional interacting spin system on $\Z^d$. In this section we will present the algebraic formalism and the basic facts for the description of the kinematics and the dynamics of such systems. The interested reader can refer to \cite{Bratteli-Robinson-2} for a more complete presentation of the subject.

\subsection{Local and quasi-local observables}
  To each site $j\in\Z^d$, we associate a finite-dimensional $C^*$-algebra $\bb{A}_j$ isomorphic to ${\rm Mat}_{n_j} (\C)$, the algebra of $n_j\times n_j$ matrices with complex entries (equipped with its natural operations and norm). We do not require the dimensions $n_j$ to be uniformly bounded, but we shall assume that $\log(n_j)$ grows at most polynomially with $\max_{1 \leq i \leq d} |j_i|$.

\medskip

 Let $\s{P}(\Z^d)$ denote the power set of $\Z^d$, and let $\s{P}_0(\Z^d)\subset \s{P}(\Z^d)$  denote the collection of finite subsets.
  To each finite region $\Lambda \in \s{P}_0(\Z^d)$, we associate the local algebra
  \[
  \bb{A}(\Lambda) \;:=\;\bigotimes_{j\in\Lambda}\bb{A}_j\;.
  \] 
  If $\Lambda\subset\Sigma$, there is a natural inclusion $\bb{A}(\Lambda) \hookrightarrow\bb{A}(\Sigma)$ obtained by extending the elements of $\bb{A}(\Lambda)$ by the identity outside of $\Lambda$.
 Since $\s{P}_0(\Z^d)$ endowed with the inclusion is a directed set, one can construct  the inductive limit 
  \[
  \bb{A}_{\rm loc}\;:=\;\varinjlim_{\Lambda\in \s{P}(\Z^d)}\bb{A}(\Lambda)\;.
  \]
  The  $C^*$-algebra $\bb{A}$ is defined as a completion of the $\ast$-algebra $\bb{A}_{\rm loc}$, \ie
\[
\bb{A}\;:=\;\overline{ \bb{A}_{\rm loc}}^{\;\|\;\|}\;.
\]  
 Elements of $\bb{A}_{\rm loc}$ are referred to as \emph{local} observables, while elements of the completion $\bb{A}$ are the \emph{quasi-local} observables.  
  
\medskip

The lattice $\Z^d$ is equipped with the Euclidean metric. In fact, all the results of this paper hold for a more general discrete metric space $(\Gamma,d)$ provided the following holds:
		\begin{enumerate}
	\item There exist $C_d > 0$ such that for any $r \in \N$,
		\begin{equation}\label{VolumenBound}
			\sup_{j \in \Gamma} \left|B_j(r + 1)\right| \leqslant C_d r^{d},
		\end{equation}
			where $B_j(r) \; := \; \{i \in \Gamma: \; d(j,i) \leqslant r - 1\}$ and $B_j(0) = \emptyset$.
			\vspace{1mm}
			
		\item There exists $K_d > 0$ such that for any $r \in \N$,
		\begin{equation}\label{BoundaryBound}
			\sup_{j \in \Gamma} \left|\partial B_j(r + 1)\right| \leqslant K_d r^{d -1},
		\end{equation}
		where $\partial B_j(r) \; := \; B_j(r) - B_j(r - 1)$.

\end{enumerate}
\subsection{Almost-local observables}\label{sub:ALO}

Evidently, $\bb{A}_{\rm loc}$ is a dense $\ast$-subalgebra of $\bb{A}$ of \virg{compactly supported} observables. It is useful to consider other algebras that sit between $\bb{A}_{\rm loc}$ and $\bb{A}$ and are defined by the rate of convergence of the local approximations. To formalize this, for any  $A\in \bb{A}$ and $j \in \Z^d$, let us consider the function
  \begin{equation}\label{Seminorms fj}
  f_{j,r}(A)\;:=\;\inf_{B\in \bb{A}({B_j(r))}}\;\|A-B\|\;.
  \end{equation}
  with the convention that $f_{j,0}(A)=\inf_{c\in\C}\|A-c{\bf 1}\|$. In \cite{Kapustin-Sopenko-Local-Noether}, the authors discuss some properties of this function. In particular, this is a monotonically decreasing non-negative function of $r$ which takes values in $\left[0, \|A\|\right]$ and approaches zero as $r\to\infty$. It is, moreover, subadditive: for every fixed $r\geqslant 0$ and  $j\in\Z^d$ one has that
  \[
   f_{j,r}(A_1+A_2)\;\leqslant\;  f_{j,r}(A_1)\;+\;  f_{j,r}(A_2)
  \]
  Thus, the $f_{j,r}$ provides a family of seminorms on $\bb{A}$.
  The fact that they are only seminorms and not norms follows by observing that when $A\in \bb{A}(\Lambda)$ is a local observable then   $f_{j,r}(A)=0$
  for every  $r$ such that $\Lambda\subseteq B_j(r)$.
 
 \medskip

A sequence $g:\N_0\to\C$ is in 
 $\s{S}(\N_0)$ if
 \[
|||g|||_k\;:=\;\sup_{r\in\N_0}\;(1+r)^k|g(r)|\;<\;+\infty\;,\qquad \forall\;k\in\N_0
\] 
 and  $\s{S}(\N_0)$ turns out to be a  a Fr\'echet space when topologized by the family  of norms $|||\;|||_k$ (see Appendix \ref{app:Rap_seq}). Since $(1+r)^{k+1}\geqslant (1+r)^k$
 it follows that $|||g|||_k\leqslant |||g|||_{k+1}$ for every $g\in\s{S}(\N_0)$, meaning that the system of norms is increasing.
  The subset of monotonically decreasing non-negative sequence in $\s{S}(\N_0)$ will be denoted by $\s{S}^+(\N_0)$.
Given a  $g\in \s{S}(\N_0)$ we will say that $A\in\bb{A}$ is $g$-localized around $j$ if $f_{j,r}(A)\leqslant g(r)$ for all $r\in \s{S}(\N_0)$. In this case one has that the sequence $r\mapsto f_{j,r}(A)$ is in $\s{S}^+(\N_0)$. Evidently the role of the special point $j$ in the concept of localization play no special role. In fact, one has that
\[
f_{0,r+d(j, 0)}(A)\;\leqslant\;f_{j,r}(A)\;\leqslant\;f_{0,r-d(j,0)}(A)
\]
 for all $r\geqslant d(j,0)$. Therefore, localization around any $j\in\Z^d$ is equivalent to localization around $0$.  With this in mind we can consider only the family of seminorms $f_r:=f_{0,r}$ on $\bb{A}$. Let us define
 \[
   \bb{A}_{\rm \infty}\;:=\;\{A\in\bb{A}\;|\; r\mapsto f_r(A)\in\s{S}^+(\N_0)\}\;.
 \]
It follows immediately that
  \[
 \bb{A}_{\rm loc}\;\subset\;  \bb{A}_{\rm \infty}\;\subset\;\bb{A}\;,
  \]
  and so $\bb{A}_{\rm \infty}$ is a dense $\ast$-subspace of $\bb{A}$. Over this subspace, we can define the following family of norms
  \[
  |||A|||_k'\;:=\;\|A\|\;+\; |||f(A)|||_k\;,\qquad k\in\N_0\; .
  \]
\begin{proposition}\label{Frechet algebra properties}   
  The space $\bb{A}_\infty$ is a Fr\'echet space with respect to system of norms $\{|||\cdot|||_k':k\in\N_0\}$. Moreover it coincides with the closure of $\bb{A}_{\rm loc}$ with respect to the Fr\'echet topology. Finally,  $\bb{A}_\infty$ has the structure of a Fr\'echet $\ast$-algebra: for all $A,B \in \bb{A}_\infty$ and $k \in \N_0$,
  $$|||A \ B|||_k \leqslant \frac{3}{2} \ |||A|||_k \ |||B|||_k, \qquad |||A^*|||_k \; = \; |||A|||_k.$$
  \end{proposition}

\medskip

A short summary of the theory of Fr\'echet spaces and $\ast$-algebras is provided in Appen\-dix \ref{app:Fre_Al}. The results above are proven in \cite[Lemma 2.1 \& Proposition B.1]{Kapustin-Sopenko-Local-Noether}.
The elements of $\bb{A}_{\rm \infty}$ are called \emph{almost-local} observables.
  
\medskip
  
For practical reasons it is preferable to introduce another set of seminorms that provides the same topology for $\bb{A}_{\rm \infty}$. For that let us denote with 
   $\n{U}_j$  the group of unitary elements in the subalgebra $\bb{A}_{j}$.
  For any $j \in \Z^d$, let $\Pi_j: \bb{A} \rightarrow \bb{A}$ be the partial trace operator over the site $j$, which is given by: 
  \[
  \Pi_j(A)\;:=\;\int_{\n{U}_j}\dd\mu_{j}(U)\; UAU^*\;,
  \]
  with $\mu_{j}$ the normalized Haar measure on the group of the unitary operators $\n{U}_j$. For any finite subset $\Lambda \subseteq \Z^d$, the corresponding partial trace operator $\Pi_\Lambda: \bb{A} \rightarrow \bb{A}$ is defined as
  $$\Pi_\Lambda := \bigotimes_{j \in \Lambda} \Pi_j\;.$$
  For a general subset $\Sigma \subseteq \Z^d$, the partial trace operator is defined as the direct limit over the family of finite subsets $\Lambda \subset \Sigma$ as
  $$\Pi_\Sigma = \varinjlim_{\Lambda \subset \Sigma} \Pi_\Lambda.$$
The linear map $\Pi_\Lambda$ is positive and bounded $\|\Pi_\Lambda(A)\|\leqslant \|A\|$. Moreover, $\Pi_\Lambda(A)\in \bb{A}(\Lambda^c)$ while $\Pi_\Lambda(A)=A$ for every $A\in \bb{A}(\Lambda^c)$. For any $\Lambda_1, \Lambda_2 \subseteq \Z^d$,   the partial trace satisfies $\Pi_{\Lambda_1} \circ \Pi_{\Lambda_2} = \Pi_{\Lambda_1 \cup \Lambda_2} = \Pi_{\Lambda_2} \circ \Pi_{\Lambda_1}$. 
Applying this to $\Lambda = \Z^d$, one gets that $\Pi_{\Z^d}:\bb{A}\to\C{\bf 1}$. Therefore, $\Pi_{\Z^d}(A)=\omega_\infty(A){\bf 1}$, and the map $\omega_\infty:\bb{A}\to\C$ is, in fact, the unique tracial state on $\bb{A}$ (also called infinite temperature equilibrium state).
For more details on the partial trace the reader is referred to \cite{nachtergaele-scholz-werner-13} and \cite[Remark in p. 245]{Bratteli-Robinson-2}.
  
\medskip  
  
Using these partial trace operators, let us introduce the following family of functions:
 \[
 p_r(A)\;:=\;\|A- \Pi_{B_0(r)^c}(A)\|\;,\qquad A\in \bb{A}\;.
 \] 
In particular, $p_0(A) = \|A - \omega_\infty(A)\bf{1}\|$.
With these functions one can  define a new family of norms  over $\bb{A}_{\rm loc}$ given by:
 \begin{equation}\label{nor_001}
 \|A\|_k \;:=\; \|A\| + \sup_{r\in\N_0}\; (1 + r)^k\; p_r(A)\;.
 \end{equation}
 The next result  can be found in \cite[Proposition D.1]{Kapustin-Sopenko-Local-Noether}. Nevertheless, we will sketch the proof since it is  based on certain inequalities that will be important for the following of this work.
\begin{lemma}\label{Equivalence bounds rapidly decaimiento}
	The family of norms $\|\cdot\|_k$ and $|||\cdot|||_k'$ defined over $\bb{A}_{\rm loc}$ are equivalent. Hence, they generate the same 
	Fr\'echet $\ast$-algebra
	 $\bb{A}_{\infty}$.
\end{lemma}
\proof
Since $f_r(A)\leqslant p_r(A)$ by construction, one immediately gets 
$|||A|||_k'\leqslant \|A\|_k$ for every $A\in\bb{A}_{\rm loc}$ and $k \in \N_0$. For the other inclusion, let us observe that for any finite $\Lambda \subseteq \Z^d$, the algebra $\bb{A}(\Lambda)$ is a closed finite dimensional subalgebra of $\bb{A}$. Thus, given $A \in \bb{A}$, there exists $A_\Lambda \in \bb{A}(\Lambda)$ such that
    $$\|A - A_\Lambda\|\; =\; \inf_{B \in \bb{A}(\Lambda)} \|A - B\|\;.
    $$
    Using this, we get 
  \[
  \|A - A_\Lambda\|\; \leqslant\;\|A - \Pi_{\Lambda^c}(A)\|
  \]
since $\Pi_{\Lambda^c}(A)$ is supported in $\bb{A}(\Lambda)$. 
Using the linearity of $\Pi_{\Lambda^c}$, the fact that it acts as the identity on $\bb{A}(\Lambda)$ and the usual triangle inequality, one gets
\[
\|A - \Pi_{\Lambda^c}(A)\|\; \leqslant\;||A - A_\Lambda|| + ||\Pi_{\Lambda^c}(A_\Lambda - A)||\;\leqslant\; 2||A - A_\Lambda||
\]
where the last equality is consequence of the fact that the partial trace cannot increase the norm of an operator. This shows that
$p_r(A)\leqslant2f_r(A)$ and in turn $\|A\|_k\leqslant 2|||A|||_k'$.\qed	 

\medskip

From these bounds and Proposition \ref{Frechet algebra properties}, we obtain that for any  $A,B \in \bb{A}_\infty$ and any $k \in \N_0$
$$||AB||_k \; \leqslant \; 3 \ ||A||_k \ ||B||_k.$$
For the involution, it follows directly from the definition that $||A^*||_k = ||A||_k$. 
Moreover, similarly as we have defined the functions $f_{j,r}$, we can also consider the family of functions $p_{j,r}(A)$ defined as
$$p_{j,r}(A) := \|A - \Pi_{B_j(r)^c}(A)\|.$$
In fact, the same estimates as those presented in Lemma~\ref{Equivalence bounds rapidly decaimiento} remain valid for $f_{j,r}$ and $p_{j,r}$.

\medskip

 Following  \cite{Kapustin-Sopenko-Local-Noether}, we further introduce the subspaces $\bb{D}_\sharp \subset\bb{A}_\sharp$, where $\sharp$ stands for ${\rm loc}$ or ${\infty}$, \etc, and $\bb{D}(\Lambda) \subset \bb{A}(\Lambda)$ for $\Lambda \in \s{P}(\Z^d)$. These spaces are made by elements such that: (i) $A^*=-A$ (anti-self-adjoint), and (ii) $\omega_\infty(A)=0$ (traceless). One has inclusions $\bb{D}(\Lambda) \subset\bb{D}_{\rm loc} \subset   \bb{D}_{\infty}\subset   \bb{D}$. All these spaces are real Lie algebras with respect to the commutator. The space $\bb{D}_{\infty}$ is the completion of $\bb{D}_{\rm loc}$ with respect to one of the systems of norms in Lemma \ref{Equivalence bounds rapidly decaimiento}. Similarly $\bb{D}$ is the completion of $\bb{D}_{\rm loc}$ with respect to the $C^*$-norm.
For any $A\in\bb{D}_{\infty}$,
\[
\|A\|\;=\;p_0(A)\;\leqslant\; 2f_0(A)\;\leqslant \;2|||f(A)|||_k
\] 
where the first equality is valid for traceless operator, the second is in the proof in Lemma~\ref{Equivalence bounds rapidly decaimiento} and the last one follows from the definition of  $|||\cdot|||_k$. Then,  for every traceless operator $A$, one gets that
\[
\frac{1}{3}|||A|||_k'\;\leqslant \;|||f(A)|||_k\;\leqslant \;|||A|||_k'
\]
where the second equality follows immediately from the definition. Therefore the topology on $\bb{D}_{\infty}$ can be induced by the simpler system of norms
$$
||A||'_k\; :=\; |||f_r(A)|||_k\;=\;\sup_{r \in\N_0}(1 + r)^k\; f_r(A)\;.
$$
 It turns out that $\bb{D}_{\infty}$ is a Fr\'echet-Lie algebra.
 
\begin{remark}[Pauli basis]\label{Remark: PauliBasis}
	For any finite subset $\Lambda := \{j_1, \dots, j_r\} \subset \Z^d$, let us construct a particular basis for the real subspace  $\bb{D}(\Lambda)$. First, note that the subspace ${\rm Mat}_{n_{j_i}}^{{\rm sh}}(\C)$ of skew-Hermitian $n_{j_i} \times n_{j_i}$ matrices has \emph{real} dimension $n_{j_i}^2$ and possesses a (non-unique) basis $\{E_k^{(j_i)}\}_{k=1,\ldots,n_{j_i}^2}$ such that $E_1^{(j_i)} = \ii{\bf 1}$ and ${\rm Tr}(E_k^{(j_i)}) = 0$ for $2 \leqslant k \leqslant  n_{j_i}^2$. We will refer to such matrices as \emph{(generalized) Pauli matrices}. Hence, the collection
$$
\left.\left\{E_{k_1}^{(j_1)} \otimes E_{k_2}^{(j_2)} \otimes \dots \otimes E_{k_r}^{(j_r)} \;\right|\; (k_1, \dots, k_r) \neq (1, \dots, 1)\right\}
$$
forms a basis for $\bb{D}(\Lambda)$. We will refer to any basis of this form as a \emph{(generalized) Pauli basis}.
 \hfill $\blacktriangleleft$
\end{remark}
 
\subsection{Automorphisms and Fr\'echet-type automorphisms}

A classical result in the theo\-ry of $C^*$-algebras is that every $*$-homomorphism $\pi$ between two $C^*$-algebras $\bb{B}$ and $\bb{C}$ is norm-decreasing, with equality if and only if $\pi$ is one-to-one, see \cite[Section 2.3]{Bratteli-Robinson-1}. Hence, if $\pi$ admits an inverse $\pi^{-1}$, then this inverse is automatically continuous. We refer to any invertible $*$-homomorphism $\alpha: \bb{A} \rightarrow \bb{A}$ as an automorphism, and denote the set of all automorphisms of $\bb{A}$ by ${\rm Aut}(\bb{A})$. Clearly, ${\rm Aut}(\bb A)$ is a group under composition.
In analogy, a $*$-homomorphism $\alpha:\bb{A}_\infty\to \bb{A}_\infty$ is a map which respects the $\ast$-algebra structure of $\bb{A}_\infty$. It is continuous if and only if for every $k\in\N_0$ there is a $m\in \N_0$ and a constant $C_k>0$ such that $\|\alpha(A)\|_k\leqslant C_k\|A\|_m$ for every $A\in\bb{A}_\infty$. A $*$-homomorphism $\alpha$ is called an automorphism of $\bb{A}_\infty$ if it is a continuous $*$-homomorphism with a continuous inverse $\alpha^{-1}$. The group of continuous automorphisms of $\bb{A}_\infty$ is denoted by ${\rm Aut}(\bb{A}_\infty)$. 

%

Before discussing the relation between ${\rm Aut}(\bb A_\infty)$ and ${\rm Aut}(\bb A)$, we introduce two fami\-lies of automorphisms of $\bb A$ that will play a role later on.

%
\begin{example}[Flip-type automorphisms]\label{ex:flip_a}
We start with a class of 1-dimensional automorphisms which will be relevant in the sequel. The spin chain is assumed to have isomorphic algebras for all sites, namely $\bb{A}_j \simeq {\rm Mat}_{n} (\C)$ for every $j\in\Z$.
Consider a label function $\zeta:\Z\to\Z$ such that: i) it is strictly increasing (hence injective), and ii)
$\Delta_\zeta(j):=\zeta(j+1)-\zeta(j)\geqslant 2$.
Therefore, for each $j\in\Z$ one has a unique pair $\zeta(j)$ and $\zeta(j+1) - 1$ and one can consider the map $\psi_\zeta: \bb{A} \rightarrow \bb{A}$ initially defined on simple tensors by exchanging the elements at sites $\zeta(j)$ and $\zeta(j+1) - 1$ for all {$j \in \Z$}, keeping the others fixed. More precisely, one has that $\psi_\zeta$ sends the element
\begin{align*}
\ldots \otimes A_{\zeta(j - 1)} \otimes \ldots \otimes A_{\zeta(j) - 1} \otimes A_{\zeta(j)} \otimes A_{\zeta(j) + 1} \otimes A_{\zeta(j) + 2} \otimes \ldots \otimes A_{\zeta(j + 1) - 1} \otimes \ldots 
\end{align*}
to 
\begin{align*}
\ldots \otimes A_{\zeta(j) - 1} \otimes \ldots \otimes A_{\zeta(j - 1)} \otimes  A_{\zeta(j + 1) - 1} \otimes A_{\zeta(j) + 1} \otimes A_{\zeta(j) + 2} \otimes \ldots \otimes A_{\zeta(j)} \otimes \ldots\;.
\end{align*}
This map extends linearly to any local observable $A\in \bb{A}_{\rm loc}$. By construction, $\psi_\zeta$ defines a $*$-homomorphism on $\bb{A}_{\rm loc}$. In fact, verifying that $\psi_\zeta(AB)=\psi_\zeta(A)\psi_\zeta(B)$ and $\psi_\zeta(A^*)=\psi_\zeta(A)^*$ is straightforward on monomials. Moreover, one gets that $\psi_\zeta^2 = {\rm id}$, which shows that $\psi_\zeta$ is involutive. It turns out that $\psi_\zeta$ can be extended to an automorphism of $\bb{A}$. First of all, notice that for any finite subset $\Lambda \subset \Z$, $\psi_\zeta$ can be restricted to an $*$-homomorphism $\psi_\zeta|_\Lambda: \bb{A}(\Lambda) \rightarrow \bb{A}_{\rm loc} \subset \bb{A}$. Since $\psi_\zeta|_\Lambda$ is injective, it follows that $\|\psi_\zeta|_\Lambda (A)\| = \|A\|$ for every finite $\Lambda$, and in turn $\|\psi_\zeta(A)\| = \|A\|$ for all $A \in \bb{A}_{\rm loc}$. Hence, by density, $\psi_\zeta$ extends to a $\ast$-automorphism $\psi_\zeta: \bb{A} \rightarrow \bb{A}$ such that  $\psi_\zeta^2 = {\rm id}$ still holds. Automorphisms of this type will be called \emph{generalized flips}. For instance,
the simple case $\zeta(j)=2j$ corresponds to the flip-automorphism which exchanges observables between (neighbouring) even and odd sites.
 \hfill $\blacktriangleleft$
\end{example}

\begin{example}[Shift-type automorphisms]\label{ex:shif_a}
As in the previous example, let us construct ano\-ther family of automorphisms on the one-dimensional spin lattice system $\bb{A}$ on $\Z$ which will be relevant in the sequel. Consider a strictly increasing (hence injective) function $\xi: \Z \rightarrow \Z$, and consider the map $\sigma_\xi: \bb{A} \rightarrow \bb{A}$ defined on simple tensors by shifting the elements at the site $\xi(j)$ to the site $\xi(j+1)$ for all $j\in\Z$. More precisely, one has that $\sigma_\xi$ sends
\begin{align*}
\ldots \otimes A_{\xi(j - 1)} \otimes \ldots \otimes A_{\xi(j)} \otimes A_{\xi(j) + 1} \otimes A_{\xi(j) + 2} \otimes \ldots \otimes A_{\xi(j + 1) - 1} \otimes \ldots 
\end{align*}
to
\begin{align*}
\ldots \otimes A_{\xi(j - 2)} \otimes \ldots \otimes A_{\xi(j - 1)} \otimes  A_{\xi(j) + 1} \otimes A_{\xi(j) + 2} \otimes \ldots \otimes A_{\xi(j + 1) - 1} \otimes A_{\xi(j)} \otimes \ldots\;.
\end{align*}
The argument which shows that $\sigma_\xi$ extends to  a 
$\ast$-automorphism of  $\bb{A}$ is similar to the one used in Example \ref{ex:flip_a}. Automorphisms of this type will be called \emph{generalized shifts}. The simplest case $\xi(j)=j$ correspond to the usual shift of the full spin chain.
\hfill $\blacktriangleleft$
\end{example}
\begin{remark}\label{Remark: Ex From Z to Z^d}
Both families of automorphisms, $\psi_\zeta$ and $\sigma_\xi$, introduced in Examples \ref{ex:flip_a} and \ref{ex:shif_a} respectively, were initially defined for the one-dimensional spin lattice $\Z$. However, these constructions can be naturally extended to higher-dimensional lattices. To do so, it is enough to identify $\Z$ with the sublattice $\Z \times \{0\} \times \cdots \times \{0\} \subset \Z^d$ and define the action of $\psi_\zeta$ and $\sigma_\xi$ on this sublattice as before, while letting them act trivially (i.e., as the identity) on the remainder of the lattice. 
\hfill $\blacktriangleleft$
\end{remark}	  
%
%

We now turn to the relation between automorphisms of the Fr\'echet algebra $\bb{A}_\infty$ and those of the algebra $\bb{A}$.
\begin{proposition}\label{Frechet continuity implies Operator continuity}
Let $\alpha: \bb{A}_\infty \rightarrow \bb{A}_\infty$ be an automorphism of the Fr\'echet algebra $\bb{A}_\infty$. Then, there exists a unique continuous extension $\widetilde{\alpha}: \bb{A} \rightarrow \bb{A}$, with $\widetilde{\alpha} \in {\rm Aut}(\bb{A})$.    
\end{proposition}
\proof
	First, observe that for any $\Lambda \in \s{P}_0(\Z^d)$, the $*$-homomorphism $\alpha$ can be restricted to the subalgebra $\bb{A}(\Lambda)$. Using the inclusion map $\bb{A}_\infty \xhookrightarrow{} \bb{A}$, this restriction induces a $*$-homomorphism of $C^*$-algebras $\alpha|_\Lambda: \bb{A}(\Lambda) \rightarrow \bb{A}$. 
Therefore, one has that $\|\alpha(A)\|=\|\alpha|_\Lambda(A)\|
\leqslant\|A\|$ for every $A\in \bb{A}(\Lambda)$. Since this is true for all finite $\Lambda$, one concludes that $\|\alpha(A)\|
\leqslant\|A\|$ for every $A\in \bb{A}_{\rm loc}$. The continuity and the density of 	$\bb{A}_{\rm loc}$ implies that $\alpha$ can be uniquely extended to a global $*$-homomorphism $\widetilde{\alpha}: \bb{A} \rightarrow \bb{A}$.

It remains to prove that $\widetilde{\alpha}|_{\bb{A}_\infty} = \alpha$. Let $A \in \bb{A}_\infty$, and let $\{A_n\}_{n \in \N}$ be a sequence in $\bb{A}_{\rm loc}$ converging to $A$ in the Fr\'echet topology. By continuity of $\alpha$, for any $k \in \N_0$,
\[
\lim_{n\to\infty}\|\alpha(A) - \alpha(A_n)\|_k \;=\;0\;.
\]
 Since $\|\cdot\| \leqslant \|\cdot\|_k$, one concludes that
    \[
    \lim_{n\to\infty}\|\alpha(A) - \widetilde{\alpha}(A_n)\| \;=\;0
    \]
From the inequality
\[
\|\alpha(A) - \widetilde{\alpha}(A)\|\;\leqslant\;\|\alpha(A) - \widetilde{\alpha}(A_n)\|+\|\widetilde{\alpha}(A_n)- \widetilde{\alpha}(A)\|
\]    
and the continuity of $\widetilde{\alpha}$ one infers that  $\alpha(A) = \widetilde{\alpha}(A)$.

The same argument works for the inverse $*$-homomorphism $\alpha^{-1}$ which can be extended to a unique 
	$*$-homomorphism $\widetilde{\alpha'}: \bb{A} \rightarrow \bb{A}$ which coincides with $\alpha^{-1}$ over $\bb{A}_\infty$.
	It remains to show that $\widetilde{\alpha'}$ is the inverse of $\widetilde{\alpha}$. For that, observe that $(\widetilde{\alpha'}\circ\widetilde{\alpha})(A)=\widetilde{\alpha'}(\alpha(A))= {\alpha^{-1}}(\alpha(A))=A$ for every $A\in \bb{A}_{\rm loc}$
where we used that $\alpha(A)\in \bb{A}_\infty$. Then $\widetilde{\alpha'}\circ\widetilde{\alpha}$ acts as  the identity in 
$\bb{A}_{\rm loc}$ and by density $\widetilde{\alpha'}\circ\widetilde{\alpha}={\rm Id}_{\bb{A}}$. A similar argument shows that $\widetilde{\alpha}\circ\widetilde{\alpha'}={\rm Id}_{\bb{A}}$
which implies that $\widetilde{\alpha'}=\widetilde{\alpha}^{-1}$.
\qed

\medskip

The result above states that any Fr\'echet automorphism admits a unique \emph{lift} to an automorphism of the full algebra.  
We can express this result with the symbol $\imath:{\rm Aut}(\bb{A}_\infty)\hookrightarrow {\rm Aut}(\bb{A})$ where $\imath$ is the \emph{lift map} described above. In the next section we will study under which condition an element in ${\rm Aut}(\bb{A})$ induces a Fr\'echet automorphism in ${\rm Aut}(\bb{A}_\infty)$. 

\section{Restriction of automorphisms to the Fr\'echet $\ast$-algebra}
In this section we want to explore the opposite direction of Proposition~\ref{Frechet continuity implies Operator continuity}. More precisely, we want to find conditions that allow to restrict an automorphism $\alpha \in {\rm Aut}(\bb{A})$ to a  Fr\'echet automorphism of the $\ast$-algebra $\bb{A}_\infty$.  
For this purpose we will go through the notions of  \emph{almost locality preserving automorphism} and \emph{Lieb-Robinson bound condition}.

\subsection{Almost locality preserving automorphisms}\label{sec-ALP}
The following discussion was first introduced in \cite{Kapustin-Sopenko-Lieb} in the context of one-dimensional spin lattice systems. It provides a sufficient condition under which automorphisms of the quasi-local algebra $\bb{A}$ restrict to automorphisms of the Fr\'echet algebra $\bb{A}_\infty$. 

To begin, given an automorphism $\alpha \in {\rm Aut}(\bb{A})$, we introduce the following non-negative quantity:
\begin{equation}\label{Function H_alpha}
	 H_\alpha (s, r)\; := \;\sup_{p \in \Z^d } \left(\sup_{\substack{A \in \bb{A}(B_p(s)) \\ ||A|| = 1}} \left(\inf_{B \in \bb{A}(B_p(s + r))} ||\alpha(A) - B||\right)\right)\;,
\end{equation}
defined for $s, r \geqslant 1$. For convenience, we set $H_\alpha(s, 0) := 1$ for all $s \geqslant 1$. 
\begin{definition}[Almost locality preserving automorphisms]\label{ALP automorphism firts definition}
We will say that an automorphism $\alpha \in {\rm Aut}(\bb{A})$ is \emph{local} if there exists a $R$ such that $ H_\alpha (s, r)=0$ for every $r>R$ and $s \geqslant 1$; in that case $R$ is called the range of locality. We will say that $\alpha$ is \emph{almost locality preserving} if there exists a function $f^{(\alpha)} \in \s{S}(\N_0)$ such that
 $$H_\alpha (s,r) \leqslant s^{d-1} f^{(\alpha)} (r).$$
\end{definition}
According to the discussion in Appendix \ref{app:Rap_seq}, one can find a function $\Tilde{f}^{(\alpha)} \in \s{S}^+(\N_0)$ such that $f^{(\alpha)} (r) \leqslant \Tilde{f}^{(\alpha)}$ for all $r$. Hence, without lossing of generality, we will say that an automorphism $\alpha$ is almost locality preserving if the function $f^{(\alpha)}$ belongs to $\s{S}^+(\N_0)$.

\medskip

We will denote by ${\rm L\text{-}Aut}(\bb{A})$ the set of local  automorphisms (of any range). If $\alpha$ is a local automorphism of range $R$, then for any $A \in \bb{A}(B_p(s))$ it follows that $\alpha(A)\in \bb{A}(B_p(s + R))$ for every $s \geqslant 1$. Therefore $\alpha$ maps local observables to local observa\-bles. These automorphisms are also known in the literature as quantum cellular automata (QCA), see in particular~\cite{GrossNesmeVogtsWerner}.
We will denote by ${\rm ALP\text{-}Aut}(\bb{A})$
the set of almost locality preserving  automorphisms
of $\bb{A}$. In this case, $\alpha$ maps local observables to quasi-local ones but with ``tails'' that decay rapidly in space. Clearly ${\rm L\text{-}Aut}(\bb{A})\subset {\rm ALP\text{-}Aut}(\bb{A})$. To better understand the structure and implications of this class of automorphisms, we first present a simple lemma that will be useful in analyzing their proper\-ties. We then provide two illustrative examples: one satisfying the ALP condition and another that explicitly violates it.

\begin{lemma}\label{lemma: ALP implies Diagonal}
	Let $\alpha \in \operatorname{Aut}(\bb{A})$ be an almost locality preserving automorphism. Let $D^{(\alpha)}: \N_0 \rightarrow \R_{\geqslant 0}$ be the function defined by
	$$
	D^{(\alpha)} (r) := 
	\begin{cases} 
		H_{\alpha} (r,r) & \text{if } r \geqslant 1, \\
		1 & \text{if } r = 0.
	\end{cases}
	$$
	Then, $D^{(\alpha)} \in \s{S}(\N_0)$.
\end{lemma}
\begin{proof}
	It is an straightforward observation. From Definition~\ref{ALP automorphism firts definition}, one has that
	$$D^{(\alpha)}(r) \; =\; H_\alpha (r,r) \leqslant r^{d-1} f^{(\alpha)} (r).$$
	Since $f^{(\alpha)}$ belongs to $\s{S}^+(\N_0)$, it follows that the function $D^{(\alpha)}$ belongs to $\s{S}(\N_0)$. Furthermore, from Appendix \ref{app:Rap_seq}, we can assume without loss of generality that $D^{(\alpha)} \in \s{S}^+(\N_0)$.
\end{proof}
\begin{remark}

It is important to mention the physical interpretation of the bound in Definition~\ref{ALP automorphism firts definition}. The ``tail'' of the evolved operator $\alpha(A)$ outside $B_p(s+r)$ could be interpreted as a boundary effect. It is primarily due to the portion of $A$ supported near the surface of $B_p(s)$.

\noindent The $s^{d-1}$ prefactor reflects this, as it scales with the surface area of the ball $B_p(s)$. Because of this polynomial factor in $s$, the estimate $H_\alpha (s,r) \leqslant s^{d-1} f^{(\alpha)} (r)$ is generally loose (or trivial) when the distance $r$ is small relative to the radius $s$. 

\noindent The definition is most meaningful in the large $r$ regime, where the rapid decay of $f^{(\alpha)}(r)$ must be strong enough to overcome this $s^{d-1}$ surface-area factor. Lemma \ref{lemma: ALP implies Diagonal} provides a first, non-trivial check of this condition. It confirms that the decay is strong enough to control the bound even in the ``diagonal'' case $r=s$, where the distance of interest scales precisely with the radius of the ball.
\hfill $\blacktriangleleft$
\end{remark}

\begin{example}
Examples \ref{ex:flip_a} and \ref{ex:shif_a} provide examples of automorphisms in ${\rm L\text{-}Aut}(\bb{A})$. We consider the flip-type $\ast$-automorphism $\psi_\zeta: \bb{A} \rightarrow \bb{A}$ induced by the function  $\zeta(j) = Cj$ for a given natural number $C\geqslant 2$. For any $p \in \Z$, $s \geqslant 1$ and $A \in \bb{A}(B_p(s))$, it follows that $\psi_\zeta(A)$ is supported at most in $B_p(s + C-1)$. Hence, for $r > C-1$, the function $H_{\psi_\zeta}(s, r)$ vanishes.
This shows that $\psi_\zeta$ is a local automorphism with  range $R=C-1$. The shift-type automorphisms induced by label functions of the form $\xi(j)=Cj$  for a given natural number $C\geqslant 1$ are also local of finite range. \hfill $\blacktriangleleft$
\end{example}
\begin{example}\label{Example Not ALP}
With an adequate choice of label function, flip-type automorphisms can also violate the ALP-condition. Let $\zeta: \Z \rightarrow \Z$ be a strictly increasing polynomial function of degree at least 2. Then the difference $\Delta_\zeta(j)=\zeta(j + 1)  - \zeta(j)$ grows at least linearly with $|j|$ when $j\to\pm\infty$. Furthermore, there exists $j_\ast \in \Z$ such that $\zeta(j_\ast) \leq 0$ and $\zeta(j) > 0$ for all $j > j_\ast$. This implies that the automorphism $\psi_\zeta$ {spreads} the support $[\zeta(j_\ast)-1, \zeta(j)] \cap \Z$ into the larger support $[\zeta(j_\ast-1), \zeta(j + 1) - 1] \cap \Z$.

\noindent With these observations in mind, let us show that, for any $s \geqslant 1$ and $r \geqslant 0$, the function $H_{\psi_\zeta}(s, r)$ in Definition \ref{ALP automorphism firts definition} is bounded below by $1/2$. 
Fix $s \geqslant 1$ and take $r \in \N_0$. Since $\Delta_\zeta(j)$ grows at least linearly, there exists a $j' \geqslant j_\ast$ (depending on $s$ and $r$) such that $s + r <  \Delta_\zeta(j') - 1$. 
Pick  $A \in \bb{A}(\{\zeta(j')\}))$ such that ${\rm Tr}(A)=0$ and $||A|| = 1$. Consequently, $\psi_\zeta(A)\in \bb{A}(\{\zeta(j' + 1) - 1\})$. 

\noindent In order to simplify the notation, let us introduce $B_{s,r}:=B_{\zeta(j')}(s + r)$, the ball of center $\zeta(j')$ and radius $s+r$. We now observe that
\[
\begin{aligned}
	\inf_{B \in \bb{A}\left(B_{s,r}\right)} \|\psi_\zeta(A) - B\| \;&\geqslant\;  \frac{1}{2} \|\psi_\zeta (A) - \Pi_{B_{s,r}^c}(\psi_\zeta (A))\|\;,
\end{aligned}
\]
where the inequality is a restatement of $p_r(A)\leqslant 2f_r(A)$ proved in Lemma \ref{Equivalence bounds rapidly decaimiento}.
Since $\psi_\zeta(A)$ is supported in $\{\zeta(j' + 1) - 1\}\subset B_{s,r}^c$ one gets that 
\[
\Pi_{B_{s,r}^c}(\psi_\zeta (A))\; =\; {\rm Tr}(A)\;=\;0,
\]
and, in turn
\[
\inf_{B \in \bb{A}\left(B_{s,r}\right)} \|\psi_\zeta(A) - B\| \;\geqslant\; \frac{1}{2} \|\psi_\zeta(A)||\; =\; \frac{1}{2}\;.
\]
Hence $H_{\psi_\zeta}(s, r) \geqslant 1/2$ for all $s \geqslant 1$ and $r \geqslant 0$, and the automorphism is not in ${\rm ALP\text{-}Aut}(\bb{A})$.
The simplest example that meets the conditions above is the function $\zeta(j)=2{\rm sgn}(j)\;j^2$, with $j_*=0$ and $j'>\max\{(s + r - 1)/4,0\}$.
  \hfill $\blacktriangleleft$
\end{example}
\begin{remark}\label{Remark: Ex From Z to Z^d not ALP}
As pointed out in Remark \ref{Remark: Ex From Z to Z^d}, the automorphism in Example \ref{Example Not ALP} can be extended trivially to the spin lattice $\Z^d$, yielding examples of maps that do not satisfy the ALP condition in any dimension.
\hfill $\blacktriangleleft$
\end{remark}
The next result, that generalizes \cite[Lemma A.2]{Kapustin-Sopenko-Lieb},  shows that almost-locality preserving automorphisms of $\bb{A}$ behave well with respect to the Fr\'echet topology of $\bb{A}_\infty$. 
\begin{proposition}\label{Proposicion Continuidad ALP en Frechet}
    Let $\alpha \in {\rm ALP\text{-}Aut}(\bb{A})$. Then $\alpha(\bb{A}_\infty)\subseteq \bb{A}_\infty$ and the restriction
    $\alpha|_{\bb{A}_\infty}$
is  a continuous $\ast$-homomorphism in the Fr\'echet topology of $\bb{A}_\infty$.
\end{proposition}
\proof 
We begin by recalling the function $D^{(\alpha)}: \N_0 \rightarrow \R_{\geqslant 0}$ introduced in Lemma \ref{lemma: ALP implies Diagonal}. From the proof of that lemma (and Appendix \ref{app:Rap_seq}), we can assume without loss of generality that $D^{(\alpha)} \in \s{S}^+(\N_0)$.

\noindent Let $A \in \bb{A}_\infty$. To simplify the notation, let $A|_r := \Pi_{B_0(r)^c}(A)$, with $r \in \N_0$. From the definition of $p_r(A)$ one gets
$
\|A|_r\|-\|A\|\;\leqslant\;\big|||A|_r||-||A||\big|\;\leqslant\;p_r(A)
$
and, in turn,
\begin{equation}\label{eq;norm of a pr}
||A|_r|| \leq p_r(A) + ||A||.
\end{equation}
Similarly, for any $r \in \mathbb{N}_0$, one has that   
    \[
    \begin{aligned}
    \big|p_{2r}(\alpha(A)) - p_{2r}(\alpha(A|_r))\big|\; &\leq\; p_{2r} (\alpha(A) - \alpha(A|_r))\; \leq\; 2 \|\alpha(A - A|_r)\|\\
     &=\; 2 \|A - A|_r\|\; =\; 2p_r(A)
    \end{aligned}
    \]
where the first inequality is obtained from the triangle inequality
along with the definition of $p_{2r}$, and
the second inequality follows from the the bound obtained in proof of Lemma \ref{Equivalence bounds rapidly decaimiento}.    
    Thus, one gets that
    \begin{equation}\label{eq;mm0}
    p_{2r}(\alpha(A))\; \leq\; p_{2r}(\alpha(A|_r)) + 2p_r(A)
    \end{equation}
        for every  $r\in\N_0$. The next task is to bound the first term on the right-hand side above. From the almost-locality preserving property, by fixing the ball $B_0(r)$ where the operator $A|_r$ is supported, one gets the  bound
    \begin{equation}\label{eq;mm1}
    p_{2r}(\alpha(A|_r))\; \leqslant\; 2\inf_{B \in \bb{A}(B_{0}(2r))} \|\alpha(A|_r) - B\|\; \leq\; 2D^{(\alpha)}(r)\; \|A|_r\|\;,
    \end{equation}
    where we used again the proof of Lemma \ref{Equivalence bounds rapidly decaimiento}.    Therefore, one infers from \eqref{eq;norm of a pr}, \eqref{eq;mm0} and \eqref{eq;mm1} that
\begin{equation}\label{eq;mm2}
	p_{2r}(\alpha(A)) \leqslant 2p_r(A) (D^{(\alpha)}(r) + 1) + 2D^{(\alpha)}(r) \|A\|\;.
\end{equation}
The same bound holds $p_{2r + 1}(\alpha(A))$, hence for all $r$, since $r\mapsto p_{2r + 1}(\alpha(A))$ is monotoni\-cally decreasing. Now $A\in\bb{A}_\infty$ implies that $p(A)\in\s{S}^+(N_0)$, and $\alpha\in {\rm ALP\text{-}Aut}(\bb{A})$ implies that $D^{(\alpha)}\in\s{S}^+(N_0)$. Hence $\alpha(A)\in \bb{A}_\infty$. 
 
It remains to prove the continuity. Since \[
 \sup_{r\in\N_0}\;(1+r)^k p_r(\alpha(A))\;\leqslant\;\sup_{r\in\N_0}\;(2 + 2 r)^k p_{2r} (\alpha(A)) \; ,
 \]
 for every $k\in\N_0$, one has
 \[
  \sup_{r\in\N_0}\;(1+r)^k p_r(\alpha(A))\; \leq\; 2^{k+1} \left(|||D^{(\alpha)}|||_k\|A\|+\sup_{r\in\N_0}\;(1+r)^k p_{r}(A)(1+|||D^{(\alpha)}|||_k)\right)\;.
 \]
By using the definition \eqref{nor_001} one gets
\begin{equation}\label{eq;mm3}
	\|\alpha(A)\|_k\;\leqslant\;C_\alpha(k) \|A\|_k\;,\qquad k\in\N_0
\end{equation}
 with constant $C_\alpha(k):=2^{k+1}(1+|||D^{(\alpha)}|||_k)$. This implies the continuity of $\alpha$ with respect to the Fr\'echet topology.
 \qed

\medskip

Although Proposition \ref{Proposicion Continuidad ALP en Frechet} tells us that an almost locality-preserving automorphism $\alpha$ of $\bb{A}$ can be restricted to a continuous homomorphism in the Fr\'echet $\ast$-algebra $\bb{A}_\infty$, one cannot conclude yet that this homomorphism is indeed an automorphism. For instance, the restriction $\alpha|_{\bb{A}_\infty}$ may not be surjective. We will come back to this question and answer it positively in Section~\ref{sec:restriction is automorphism}, see Corollary~\ref{corol_aut-F}.  

\medskip

First, we note that Proposition \ref{Proposicion Continuidad ALP en Frechet} establishes that ALP condition is a \emph{sufficient}  condition for Fr\'echet continuity. It turns out that it is however \emph{not necessary}, as will be illustrated with the examples of Section~\ref{Example Not ALP-2}.
\begin{proposition}\label{prop_no_go}
There is a Fr\'echet continuous automorphism $\alpha: \bb{A}_\infty \rightarrow \bb{A}_\infty$ such that $\alpha$ is not an almost locality preserving automorphism.
\end{proposition}

\subsection{Fr\'echet continuous but not ALP automorphisms} \label{Example Not ALP-2}

In Example \ref{Example Not ALP}, we observed that any flip-type $*$-automorphism $\psi_\zeta$ defined on the one-dimensional spin lattice system, with $\zeta$ a polynomial function of degree at least $2$, violates the ALP condition. We claim that the flip-type automorphisms $\psi_\zeta$ are continuous with respect to the Fr\'echet topology for all increasing polynomials $\zeta$, which proves Proposition~\ref{prop_no_go}.

\medskip

In order to prove this, let us focus on the subspace $\bb{D}\subset \bb{A}$ of skew-adjoint elements with null trace endowed  with a (generalized) Pauli basis as defined in Remark \ref{Remark: PauliBasis}. We assume that $\zeta(0) \leqslant 0$ and $\zeta(1) \geqslant 1$. Since $\psi_\zeta$ is an automorphism that squares to the identity,
$$
p_r(\psi_\zeta(A)) = \|A - \psi_\zeta ( \Pi_{B_0(r)^c} (\psi_\zeta (A)))\|
$$
for any $A \in \bb{A}$. Take $n \geqslant 1$ and $m \geqslant 0$, and define the constant $r_{m,n} := \max \; \{\zeta(n) - 1, |\zeta(-m)|\}+1$.
We will prove that
\begin{equation}\label{edD_loc}
p_{r_{m,n}}(\psi_\zeta(A)) \; = \; p_{r_{m,n}}(A)\;,\qquad \forall\; A\in \bb{D}_{\rm loc}\;. 
\end{equation}
First, consider the case where $A$ is supported on $B_0(r_{m,n})$. In this case, $\Pi_{B_0(r)^c} (A) = A$ for every $r\geqslant r_{m,n}$. Moreover, by definition, also the support of 
$\psi_\zeta (A)$ remains within $B_0(r_{m,n})$, and in turn $ \Pi_{B_0(r)^c} (\psi_\zeta (A))=\psi_\zeta (A)=\psi_\zeta (\Pi_{B_0(r)^c} (A))$. Therefore, one deduces that $p_{r}(\psi_\zeta(A)) = p_{r}(A)=0$
for every $A\in\bb{A}(B_0(r_{m,n}))$ and $r\geqslant r_{m,n}$.
Now, let $A \in \bb{D}_{\rm loc}$ be a monomial which is not strictly supported on $B_0(r_{m,n})$. Using the Pauli basis, one can write 
$$A \; = \; \bigotimes_{p \in B_0(R)} E_{k_p}^{(p)},$$
for some $R > r_{m,n}$ such that $A \in \bb{A}(B_0(R))$. Since $A$ is not supported entirely in $B_0(r_{m,n})$,  it must contain nontrivial (traceless) Pauli matrices at sites in $B_0(R) \setminus B_0(r_{m,n})$. By definition of $\psi_\zeta$, it acts as the identity or it permutes elements located in $[\zeta(n)-1, \infty)\cap \Z \times \{0\} \times \ldots \times \{0\}$ with elements in the same region, and similarly for elements in $(-\infty, \zeta(-m)-1]\cap \Z \times \{0\} \times \ldots \times \{0\}$. Hence, after applying $\psi_\zeta$, we will have a Pauli matrix located on some $q \in B_0(R) \setminus B_0(r_{m,n})$, and thus one obtains
\[
 \Pi_{B_0(r_{m,n})^c}( \psi_\zeta (A))\; =\; \Pi_{B_0(r_{m,n})^c} (A) = 0
\]
since both $A$ and $\psi_\zeta (A)$ contain trace-less matrices in $B_0(r_{m,n})^c$. Acting on both sides by $\psi_\zeta$, extending the result by  linearity and comparing with the case of observables supported inside $B_0(r_{m,n})$ one gets\[
\psi_\zeta ( \Pi_{B_0(r_{m,n})^c} (\psi_\zeta (A))) \;=\; \Pi_{B_0(r_{m,n})^c} (A)\;,\qquad \forall\; A \in \bb{D}_{\rm loc}\;.
\]
This proves \eqref{edD_loc}. 

\medskip

Now, let us fix $m =0$ and define $r_n:=r_{0,n}$. It exists a $n_*\in\N$ such that $\zeta(n_*) - 1 \geqslant |\zeta(0)|$, implying that $r_n:=\zeta(n)>0$ for every $n\geqslant n_*$. For a polynomial $\zeta$, there is $C_\zeta > 0$ such that
$$\zeta(j + 1) - \zeta(j)  \leqslant C_\zeta \zeta(j) $$
for all $j \geqslant n_\ast$, or equivalently
$r_{n+1}\leqslant (C_\zeta+1)r_n$. 
From the proof of Lemma \ref{Equivalence bounds rapidly decaimiento}, one infers that 
\[
p_{r'}(A)\;\leqslant \;2 f_{r'}(A) \;\leqslant\; 2 f_{r}(A) \;
\;\leqslant \;2 p_{r}(A) \;
\]
for any $r' \geqslant r \geqslant 0$ and any $A \in \bb{A}$. 
Therefore,
\[
\begin{aligned}
\left(1 + r_n + s\right)^k p_{r_n + s}(\psi_\zeta(A)) &\leq 2 \left(1 + r_n + s\right)^k p_{r_n}(\psi_\zeta(A)) \\
\end{aligned}
\]
for every $k\in N_0$ and $s\geqslant0$. 
Letting $s\in[0,r_{n+1}-r_n-1]$
 and using the definition of $r_n$ and the relation \eqref{edD_loc},
one can improve the inequality above as
\[
\begin{aligned}
\left(1 + r_n+ s\right)^k p_{r_n + s}(\psi_\zeta(A)) &\leq  2\left(1 + r_n+r_{n+1}-r_n-1\right)^k p_{r_n}(A)\\&=\;2\left(r_{n+1} \right)^k p_{r_n}(A)\\
& \leq 2 \left(C_\zeta+1)^k(1+r_n\right)^k p_{r_n}(A)\;.\\
\end{aligned}
\]
This implies that
\[
\begin{aligned}
\sup_{r\in[r_n, r_{n+1}-1]}\left(1 + r\right)^k p_{r}(\psi_\zeta(A))\;& \leq 2\; \left(C_\zeta+1\right)^k\left(1+r_n\right)^k p_{r_n}(A)\\
\end{aligned}
\]
for every $n\geqslant n_*$, and in turn
\begin{equation}
\sup_{r\geqslant r_{n_*}}\left(1 + r\right)^k p_{r}(\psi_\zeta(A))\; \leq\; 2  \left(C_\zeta+1\right)^k \sup_{r\geqslant r_{n_*}}\left(1+r\right)^k p_{r}(A)\;.
\end{equation}
One needs a similar bound for $0 \leqslant r \leqslant r_{n_\ast} - 1$. Since $A$ is assumed to be trace-less, then also $\psi_\zeta(A)$ is trace-less and therefore, one has
 that $p_0(\psi_\zeta(A)) =\|A\| = p_0(A)$. Therefore, for $0 \leqslant r \leqslant r_{n_\ast} - 1$ one gets that
\[
(1 + r)^k p_r(\psi_\zeta(A))\; \leq\; 2 r_{n_\ast}^k p_0(\psi_\zeta(A))\;=\;2 r_{n_\ast}^k p_0(A)\;.
\]
Setting $M_\zeta := \max\{C_\zeta + 1, r_{n_\ast}\}$, one gets the bound
$$\sup_{r \in \N_0}(1 + r)^k p_r(\psi_\zeta(A)) \; \leqslant \; 2 M_\zeta^k \sup_{r \in \N_0}(1 + r)^k p_r(A)\;$$
valid
 for any $A \in \bb{D}_{\rm loc}$ and any $k \in \N_0$,which translates into
 \begin{equation}\label{continuity shift automorphism}
 \|\psi_\zeta(A)\|_k\;\leqslant\;2 M_\zeta^k\|A\|_k\;.
 \end{equation}
Since $\bb{D}_{\rm loc}$ is dense on $\bb{D}_\infty$, we can conclude that the restriction of $\psi_\zeta$  to $\bb{D}_\infty$  is continuous with respect to the Fr\'echet topology. Finally, with Lemma \ref{Lemma: Extension}, one can conclude  that $\psi_\zeta$ is Fr\'echet continuous on the whole $\ast$-algebra $\bb{A}_\infty$.

\medskip

We conclude as in Example \ref{Example Not ALP} with the case $\zeta(j)=2{\rm sgn}(j)\;j^2$. Here, $n_* = 1$, $r_n = 2n^2$ and $C_\zeta = 3$.
\begin{lemma}\label{Lemma: Extension}
    Let $T: \bb{A}_\infty \to \bb{A}$ be a linear map satisfying (1) $T(\mathbf{1}) = 0$, or (2) $T(\mathbf{1}) = \mathbf{1}$.
    Suppose that $T(\bb{D}_\infty)\subseteq \bb{D}_\infty$ and that $T|_{\bb{D}_\infty}: \bb{D}_\infty \rightarrow \bb{D}_\infty$ is continuous with respect to the Fr\'echet topology. Then, $T: \bb{A}_\infty \rightarrow \bb{A}_\infty$ is Fr\'echet continuous.
\end{lemma}
\proof
    First, let us observe that the assertion holds on the subspace of traceless, not necessarily anti self-adjoint, elements. Since $T$ is continuous on $\bb{D}_\infty$, we have that given $m \in \N_0$, there exists a constant $C_m > 0$ and $n \in \N_0$ such that for any $B \in \bb{D}_\infty$, the following holds:
    $$|||T(B)|||_{m}' \;\leqslant\; C_m |||B|||_{n}'\;.$$
     Now, given $A \in \bb{A}_\infty$ such that $\omega_\infty (A) = 0$, we can decompose it as a linear combination of elements in $\bb{D}_\infty$ according to 
    $$A \;=\; \ii \frac{A + A^*}{2\ii} + \frac{A - A^*}{2}\;.$$
    By the linearity and continuity of $T$ on $\bb{D}_\infty$, one obtains that
    $$|||T(A)|||_{m}' \;\leqslant\;  C_m \left(\left|\left|\left|\frac{A + A^*}{2i}\right|\right|\right|_n' + \left|\left|\left|\frac{A - A^*}{2}\right|\right|\right|_n' \right) \;\leqslant\; 2 C_m |||A|||_n'\;.$$
    Now, consider the case where $T(\mathbf{1}) = 0$. Given any $A \in \bb{A}_\infty$, define its traceless part ${\tilde A} := A - \omega_\infty(A){\bf 1}$. Since $f_r(A) = f_r({\tilde A})$ for all $r$, one obtains, for any $k \in \N_0$:
    \begin{align*}
    |||\tilde A|||_{k}' \; &= \; ||\tilde A|| + \sup_{r \in \N_0}(1 + r)^k f_r(\tilde A) \;\leqslant\; \|A\| + |\omega_\infty(A)| + \sup_{r \in \N_0}(1 + r)^k f_r(A)\\
    &=\; |\omega_\infty(A)| + |||A|||_{k}' \;\leqslant\; 2 |||A|||_{k}'\;.
    \end{align*}
    Consequently, we get the bound:
    $$|||T(A)|||_{m}' = |||T(\Tilde{A})|||_{m}' \leqslant 4 C_\alpha |||A|||_{n}'.$$
     Next, consider the case where $T(\mathbf{1}) = \mathbf{1}$. Since $T(A) = T(\tilde A) + \omega_\infty(A)$ and $|||\omega_\infty(A)|||_k' = |\omega_\infty(A)| \leqslant \|A\| \leqslant |||A|||_k'$, a similar argument then yields the estimate:
     \begin{align*}
     |||T(A)|||_m'\; =\: |||T(\tilde A) + \omega_\infty(A)|||_m' \;\leqslant\; 2 C_m |||A|||_n' + |||A|||_n'\; =\; \left(2 C_m + 1\right) |||A|||_n'\;.
     \end{align*}
     This completes the proof.
\qed

\medskip

This concludes the proof of Proposition \ref{prop_no_go}. One may be lead to believe that the failure to be locality-preserving arises from the fact that the automorphisms just described are not connected to the identity. This is in fact not true and we conclude this section by providing an explicit  strongly continuous one-parameter family of continuous automorphisms of $\bb{A}_\infty$ connecting $\psi_\zeta$ to the identity.

\medskip

To begin, let us construct a one-parameter group of automorphisms $\sigma_\zeta^t$. For each $p \in \Z$, define the regions $L_{\zeta(p)} := \left([\zeta(p), \zeta(p + 1) - 1] \cap \Z\right) \times \Z \times \ldots \times \Z$ which provide a partition of $\Z^d$, and the intervals $J_{\zeta(p)} := \left([\zeta(p), \zeta(p + 1) - 1] \cap \Z\right) \times \{0\} \times \ldots \times \{0\}$. By construction, the automorphism $\psi_\zeta$ leaves invariant each subalgebra $\bb{A}(L_{\zeta(p)})$, acting non-trivially only on $J_{\zeta(p)}$, and as the identity elsewhere. Since the local algebra $\bb{A}(J_{\zeta(p)})$ is a finite-dimensional matrix algebra, it follows from \cite[Proposition 1.6]{Raeburn Morita} that the restriction $\psi_\zeta|_{\bb{A}(J_{\zeta(p)})}$ can be implemented as an inner automorphism. Moreover, because $\psi_\zeta$ acts as the identity on the rest of $L_{\zeta(p)}$, the same conclusion extends to the whole region $L_{\zeta(p)}$. That is, there exists a unitary $U_p \in \bb{A}(J_{\zeta(p)})$ such that
$$
\psi_{\zeta,p}(A)\;:=\;\psi_\zeta|_{\bb{A}(L_{\zeta(p)})}(A)\;=\;\text{Ad}_{U_p}(A)\;,\qquad A\in \bb{A}(L_{\zeta(p)}),
$$ 
where $\text{Ad}_{U}(A):=UAU^*$. Moreover, since $U_p \in \bb{A}(J_{\zeta(p)})$, we can write $U_p = \expo{itH_p}$ for some self-adjoint element $H_p \in \bb{A}(J_{\zeta(p)})$.
Hence, on each subalgebras $\bb{A}(L_{\zeta(p)})$, we define a local dynamics of the form
$$
\sigma_{\zeta,p}^t \; := \; \text{Ad}_{\expo{it H_p}}\;, \qquad t\in\R\;.
$$
Observe that $\sigma_{\zeta,p}^0={\rm Id}|_{\bb{A}(L_{\zeta(p)})}$ and $\sigma_{\zeta,p}^1=\psi_{\zeta,p}$. Since the self-adjoint elements $H_p$ commute with each other (supported in disjoint sets), we define $\sigma_\zeta^t$ on $\bb{A}_{\rm loc}$ as follows: given $A \in \bb{A}_{\rm loc}$, there exists a finite collection of regions $L_{\zeta(p_1)}, \ldots, L_{\zeta(p_r)}$ such that $A$ is supported in the union, and hence:
$$\sigma_\zeta^t(A) \; = \; e^{it(H_{p_1} + \ldots + H_{p_r})} A e^{-it(H_{p_1} + \ldots + H_{p_r})}.$$
As in Example \ref{Example Not ALP}, this family of automorphisms $\sigma_\zeta^t$ extends to the full algebra $\bb{A}$. Notice that, by construction, $\sigma_\zeta^{t + s} = \sigma_\zeta^t \circ \sigma_\zeta^s$, $\sigma_\zeta^0 = {\rm Id}$ and $\sigma_\zeta^1 = \psi_\zeta$.

\medskip

Let us now show that, for each $t \in \R$, the map $\sigma_\zeta^t$ is in fact an automorphism of the Fr\'echet algebra $\bb{A}_\infty$. Without loss of generality, assume that $\zeta(0) \leqslant 0$ and $\zeta(1) \geqslant 1$. As in the begging of this section, consider the collection $r_{m,n} = \text{max }\{\zeta(n) - 1, |\zeta(-m)|\} \ + 1$ for $m \geqslant 0$ and $n \geqslant 1$. Following the arguments provided above, in order to establish the Fr\'echet continuity of $\sigma_\zeta^t$, it suffices to show that for any $A \in \bb{D}_{\rm loc}$, the identity 
$$
p_{r_{m,n}}(\sigma_\zeta^t(A))\; =\; p_{r_{m,n}}(A)\;.
$$
holds. To do that, consider the ball $B_0(r_{m,n})$. Notice that $B_0(r_{m,n})$ is within the union of the regions $L_{\zeta(-m)}, L_{\zeta(-m + 1)}, \ldots, L_{\zeta(n-1)}$ and that $J_{\zeta(-m)}, J_{\zeta(-m + 1)}, \ldots, J_{\zeta(n-1)} \subseteq B_0(r_{m,n})\cap (\Z \times \{0\} \times \ldots \times \{0\})$

\medskip

Assume first that $A$ is supported in $B_0(r_{m,n})$. Since the automorphism $\sigma_\zeta^t$ is defined as the composition of the local automorphisms $\sigma_{\zeta,p}^t$, and the ball $B_0(r_{m,n})$ is contained in the union of disjoint regions $L_{\zeta(p)}$, it follows that $\sigma_\zeta^t(A)$ is also supported in $B_0(r_{m,n})$. Hence, for $r \geqslant r_{m,n}$, it follows that $\Pi_{B_0(r)^c}(A) = A$ and $\Pi_{B_0(r)}^c(\sigma_\zeta^t(A)) = \sigma_\zeta^t(A)$. Consequently, $p_r(\sigma_\zeta^t(A)) = p_r(A) = 0$ for all $r \geqslant r_{m,n}$.

\medskip

Now, take a monomial $A \in \bb{D}_{\rm loc}$ which is not strictly supported on $B_0(r_{m,n})$. As before, we can express this element using the Pauli basis as 
$$A \; = \; \bigotimes_{p \in B_0(R)} E_{k_p}^{(p)},$$
for some $R > r_{m,n}$ such that $A \in \bb{A}(B_0(R))$, and with nontrivial Pauli matrices at some sites $q \in B_0(R) \setminus B_0(r_{m,n})$. For those factors localized on these sites $q$, the automorphism $\sigma_\zeta^t$ acts as a composition of adjoint maps $\operatorname{Ad}_{e^{itH_l}}$ where the Hamiltonians $H_l$ are supported outside of $B_0(r_{m,n})$; hence, such elements will remain localized in $B_0(r_{m,n})^c$ under the action of $\sigma_\zeta^t$. Moreover, since $\sigma_\zeta^t$ acts as $\text{Ad}_U$ with $U$ unitary, the trace is preserved. Therefore, applying the partial trace, one obtains
$$\Pi_{B_0(r_{m,n})^c}(\sigma_\zeta^t(A)) = \Pi_{B_0(r_{m,n})^c}(A) = 0.$$ 
Extending by linearity, one gets that for all $A \in \bb{D}_{\rm loc}$
$$
\sigma_\zeta^{-t}\left(\Pi_{B_0(r_{m,n})^c}(\sigma_\zeta^t(A))\right) \;=\; \Pi_{B_0(r_{m,n})^c}(A)\;.
$$
This completes the argument for the Fr\'echet continuity of the maps $\sigma_\zeta^t$. Furthermore, the constant $2 M_\zeta^k$ in equation \eqref{continuity shift automorphism} works for all $t \in \R$.

\medskip

We now turn to the strong continuity of the one-parameter group $\sigma_\zeta^t$  with respect to the norm topology. To this end, let us consider an element $A \in \bb{A}_{\rm loc}$. Since $\sigma_\zeta^t$ is a group, it suffices to prove continuity at $t = 0$. We begin by estimating the norm difference:
\begin{align*}
	\left\| \expo{itH_p} A \expo{-itH_p} - A\right\| &\;\leqslant\; \left\| \expo{itH_p} A \expo{-itH_p} - A\expo{-itH_p}\right\| + \left\| A \expo{-itH_p} - A\right\|\\
	&\;\leqslant\; 2\|A\| \left\|\expo{itH_p} - \mathbf{1}\right\| \;\leqslant\; 2 \|A\| \left(\expo{|t| \|H_p\|}- 1\right)\;.
\end{align*}
Since $A$ is locally supported, only finitely many of the inner automorphisms $\text{Ad}_{\expo{itH_q}}$ in the definition of $\sigma_\zeta^t$ act non-trivially on $A$. Thus, by a similar computation as above, one gets:
$$\left\| \sigma_\zeta^tA) - A\right\| \;\leqslant \; 2 \|A\| \left(\expo{|t| \ \|H_{p_1} + \ldots + H_{p_r} \|}- 1\right),$$
for a finite collection $H_{p_1}, \ldots, H_{p_r}$ which, of course, depends on $A$. Now, for a general element $A \in \bb{A}$, a standard $\epsilon/3-$argument completes the proof of strong continuity with respect to the norm topology. Given $\epsilon > 0$, let $A_N \in \bb{A}_{\rm loc}$ be a locally supported element such that $\|A - A_N\| < \epsilon/3$, and let $\delta > 0$ such that for any $t$ with $|t| < \delta$. One has that $\|\sigma_\zeta^t(A_N) - A_N\| < \epsilon/3$. Then, it follows that:
$$\|\sigma_\zeta^t(A) - A\|\; \leqslant\; \|\sigma_\zeta^t(A - A_N)\| + \|\sigma_\zeta^t(A_N) - A_N\| + \|A - A_N\|\; < \;\epsilon\;.
$$

\medskip

A similar argument shows that $\sigma_\zeta^t$ is strongly continuous with respect to the Fr\'echet topology. We start with a finitely supported element $A \in \bb{A}_{\rm loc}$. Consider the following estimation:
\begin{align*}
\left\|\expo{itH_p}A\expo{-itH_p} - A\right\|_k &\leqslant \left\|\expo{itH_p}A \expo{-itH_p} - A \expo{-itH_p}\right\|_k + \left\|A\expo{-itH_p} - A\right\|_k\\
&\leqslant 3\left\|A \expo{-itH_p}\right\|_k  \left\|\expo{itH_p} - \mathbf{1}\right\|_k + 3\left\|A\right\|_k  \left\|\expo{-itH_p} - \mathbf{1}\right\|_k\\
& \leqslant  9\left(\left\|A\right\|_k  \left\|\expo{-itH_p}\right\|_k  \left\|\expo{itH_p} - \mathbf{1}\right\|_k + \left\|A\right\|_k  \left\|\expo{-itH_p} - \mathbf{1}\right\|_k \right),
\end{align*}
where the factors of $3$ arise from the multiplicative constant for the norms $||\cdot||_k$.
Since $f_r(A) \leqslant \|A\|$, it follows that $\|\expo{-itH_p}\|_k \leqslant 1 + |\zeta(p + 1)|^k$. Therefore, one concludes that
$$
\left\|\expo{itH_p}A\expo{-itH_p} - A\right\|_k\; \leqslant\; 9\|A\|_k (2 + |\zeta(p + 1)|^k) \left(e^{3|t| \ \|H_p\|_k} - 1\right)\;.
$$
With an analogous argument as for the norm case, only a finite number of $H_p$ act non-trivially on $A$, say $H_{p_1}, \ldots, H_{p_r}$. Without loss of generality, assume that $p_r$ is such that $|\zeta(p_r + 1)| > |\zeta(p_j + 1)|$. Then, it follows that
$$
\left\|\sigma_\zeta^t(A) - A\right\|_k \;\leqslant\; 9\|A\|_k (2 + |\zeta(p_r + 1)|^k) \left(e^{3|t| \ \|H_{p_1} + \ldots + H_{p_r}\|_k} - 1\right)\;.
$$
By an argument of $\epsilon/3$, we conclude the strong continuity of the one-parameter group $\sigma_\zeta^t$ with respect to the Fr\'echet topology.

\subsection{Approximately locality preserving automorphisms}\label{sec:restriction is automorphism}

We now turn to the remai\-ning claim left to prove about the restriction of elements in ${\rm ALP\text{-}Aut}(\bb{A})$, namely that it defines an automorphism of $\bb{A}_\infty$. This will follow from the fact that ${\rm ALP\text{-}Aut}(\bb{A})$ is a group. In order to prove this, we will reformulate the ALP condition. We start by recalling ideas of~\cite{Ranarad-Walter-Witteveen-Converse-Lieb}. The following definition is a restatement of \cite[Definition 2.2]{Ranarad-Walter-Witteveen-Converse-Lieb}.

\begin{definition}[$\varepsilon$-inclusion]
    Let $\bb{A}$ be a $C^*$-algebra, $\bb{B} \subset \bb{A}$ be a $C^*$-subalgebra and $\varepsilon > 0$. We say that a element $A \in \bb{A}$ is $\varepsilon$-included in $\bb{B}$ if there exists $B \in \bb{B}$ such that $||A - B|| \leq \varepsilon ||A||$. We denote this as $A {\in}_\varepsilon \bb{B}$. Similarly, given $\bb{B}$ and $\bb{C}$ $C^*$-subalgebras, we will say that $\bb{B}$ is $\varepsilon$-contained in $\bb{C}$, in symbols $\bb{B}\subset_\varepsilon\bb{C}$, if $A {\in}_\varepsilon \bb{C}$ for any $A \in \bb{B}$.
\end{definition}

In the following, we present some technical properties regarding the notion of $\varepsilon$-inclu\-sion. Both the statements and proofs have been adapted from those in \cite{Ranarad-Walter-Witteveen-Converse-Lieb}. As a notational convention, given a $C^*$-algebra $\bb{A}$ and a $*$-subalgebra $\bb{B}$, we will denote by  $\bb{B}' := \{A \in \bb{A}\;|\; \ [A,B] = 0\;,\;  \forall\; B \in \bb{B}\}$  the commutant of $\bb{B}$ in $\bb{A}$. Let us start with a general useful result.
\begin{lemma}\label{lemma_AX}
Let $\bb{A}$ be a $C^*$-algebra, and  $\bb{B}, \bb{C} \subseteq \bb{A}$  two $*-$subalgebras. If $\bb{B} \subset_\varepsilon \bb{C}'$, then 
    \[
    \|[B,C]\|\; \leq\; 2\varepsilon \|B\| \; \|C\|\;,
    \]
for any $B \in \bb{B}$ and $C \in \bb{C}$.
\end{lemma}
\proof
Let $B \in \bb{B}$, $C \in \bb{C}$ and let $A \in \bb{C}'$ satisfy $\|B - A\| \leq \varepsilon \|B\|$. Then
\[
\|[B,C]\| \;=\; \|[B - A, C]\| \;\leq\; 2 \|B - A\| \; \|C\|\; \leq\; 2\varepsilon \|B\| \; \|C\|
\]
as claimed.
\qed

\medskip

The next result specifically concerns spin algebras.
\begin{lemma}\label{Near-Inclusions}
Let $\bb{A}$ be the spin algebra associated with the lattice $\Z^d$, and let
$\alpha\in {\rm Aut}(\bb{A})$ be an automorphism of $\bb{A}$.
	\begin{enumerate}
 \item Let $\Lambda \subseteq \Z^d$ and assume that $A \in \bb{A}$ satisfies 
         $$\|[A,B]\| \leq  \varepsilon \|A\| \ \|B\|\;, \qquad\; \forall\; B \in \alpha(\bb{A}(\Lambda)).$$
     Then, $A \in_\varepsilon \alpha (\bb{A}(\Lambda^c))$.
    
    \item Let $\Lambda_i \subseteq \Z^d$, with $i=1,2$, be two subsets such that $\alpha(\bb{A}(\Lambda_i)) \subset_{\varepsilon_i} \bb{A}(\Sigma)$ for some $\varepsilon_1,\varepsilon_2>0$, and some $\Sigma \subseteq \Z^d$. Then: 
        $$\bb{A}(\Sigma^c) \subset_{\varepsilon'} \alpha(\bb{A}(\Lambda_1^c \cap \Lambda_2^c))$$
 with $\varepsilon':=2(\varepsilon_1 + \varepsilon_2)$.      
    	\end{enumerate}
    		\end{lemma}
	\begin{proof} (1) Let us start 
with  $\Lambda$  finite. 
Notice that the assumption in (1) is equivalent to $\|[\alpha^{-1}(A), D]\| \leq \varepsilon \|\alpha^{-1}(A)\|  \|D\|$ for all $D \in \bb{A}(\Lambda)$. This implies, in view of \cite[Corollary 3.1]{nachtergaele-scholz-werner-13}, that $\| \alpha^{-1}(A) - \Pi_\Lambda (\alpha^{-1}(A))\| \leq \varepsilon \|\alpha^{-1}(A)\|=\varepsilon \|A\|$.
Therefore, given $A \in \bb{A}$ as in the hypothesis, define $C := \alpha \circ \Pi_\Lambda \circ \alpha^{-1}(A)$.  Then, one has that	
\[
\|A - C\|\; =\; \|\alpha^{-1}(A) - \Pi_\Lambda(\alpha^{-1}(A))\|\; \leq\; \varepsilon\; \|A\|\;.
\]
Since by construction $C \in \alpha(\bb{A}(\Lambda^c))$, one concludes that $A \in_\varepsilon \alpha(\bb{A}(\Lambda^c))$. It remains to prove the infinite case. Take $\Lambda \subseteq \Z^d$ be an infinite subset. By construction of the spin algebra, we can find a sequence of finite subsets $\Lambda_i \subset \Lambda$ such that any element in $\bb{A}(\Lambda)$ can be approximated by an increasing  sequence of elements supported in the respective $\bb{A}(\Lambda_i)$.
Given $A \in \bb{A}$ as in the hypothesis, we can apply the result for the finite case to the elements $C_i := \alpha\circ \Pi_{\Lambda_i} \circ \alpha^{-1} (A)$. This shows that  $A \in_\varepsilon \alpha(\bb{A}(\Lambda_i^c))$ for every $\Lambda_i$. Thus, one gets
\[
\|A - \alpha ( \Pi_\Lambda ( \alpha^{-1} (A)))\|\; \leq \;\liminf_i\; \|A - C_i\| \;\leq\; \varepsilon \|A\|
\]
and this concludes the argument.\\
(2) As we have done previously, let us assume  that $\Lambda_1$ and $\Lambda_2$ are finite subsets. The infinite case can be then deduced with the same approximation procedure used above.  Let $B \in \bb{A}(\Sigma^c)$  and define  $B_1 := \alpha \circ \Pi_{\Lambda_1}\circ \alpha^{-1} (B)$ and $B_{1,2} := \alpha \circ \Pi_{\Lambda_2}\circ \alpha^{-1} (B_1)$. 
Notice that $\alpha^{-1} (B_1)$ belongs to $\bb{A}(\Lambda_1^c)$, and in turn  $B_{1,2}$ belongs to $\alpha(\bb{A}(\Lambda_1^c \cap \Lambda_2^c))$.
Then, one has
	\begin{align*}
	\|B - B_{1,2}\| &\;=\; \|\alpha^{-1}(B) - \Pi_{\Lambda_2}( \alpha^{-1}(B_1))\| = \|\alpha^{-1}(B) - \Pi_{\Lambda_2}( \Pi_{\Lambda_1} ( \alpha^{-1}(B)))\| \\
	& \leq \|\alpha^{-1} (B) - \Pi_{\Lambda_1}( \alpha^{-1}(B))\| + \|\Pi_{\Lambda_1} ( \alpha^{-1}(B)) - \Pi_{\Lambda_2}( \Pi_{\Lambda_1} ( \alpha^{-1}(B)))\|.
	\end{align*}
Let us start by examining the first term of the last inequality. From the hypothesis, one has that $\bb{A}(\Lambda_1) \subset_{\varepsilon_1} \alpha^{-1}(\bb{A}(\Sigma))$. Observing that $\alpha^{-1}(B)\in \alpha^{-1}(\bb{A}(\Sigma^c))$ and $\bb{A}(\Sigma)\subset \bb{A}(\Sigma^c)'$, one can use Lemma \ref{lemma_AX} to deduce
\[
\|[A,\alpha^{-1}(B)]\|\; \leq\; 2\varepsilon_1 \|A\| \; \|\alpha^{-1}(B)\|
\]
for any $A \in \bb{A}(\Lambda_1)$. This allows us to use \cite[Corollary 3.1]{nachtergaele-scholz-werner-13} to deduce
\begin{equation}\label{eq:LLa}
\|\alpha^{-1}(B) - \Pi_{\Lambda_1}( \alpha^{-1}(B))\|\; \leq\; 2\varepsilon_1 \|\alpha^{-1}(B)\|\;=\; 2\varepsilon_1 \|B\|\;.
\end{equation}
For the second term, by observing that $\Pi_{\Lambda_2}\circ \Pi_{\Lambda_1}=\Pi_{\Lambda_1}\circ \Pi_{\Lambda_2}$ and the fact that $\Pi_{\Lambda_1}$ reduces the norm, one gets
\[
\|\Pi_{\Lambda_1} ( \alpha^{-1}(B)) - \Pi_{\Lambda_2}( \Pi_{\Lambda_1} ( \alpha^{-1}(B)))\|\;\leqslant\; \| \alpha^{-1}(B) -   \Pi_{\Lambda_2} ( \alpha^{-1}(B))\|\;\leqslant\;2\varepsilon_2 \|B\|
\]
where the last inequality follows with the same argument used for \eqref{eq:LLa}. 
Summing up, one gets $\|B - B_{1,2}\|\leqslant \varepsilon' \|B\|$, which proves the claim.
\end{proof}

\begin{remark}\label{Generalization Near-Inclusions}
Let us observe that point (ii) in  the previous Lemma \ref{Near-Inclusions} can be generalized for any finite collection of finite subsets $\{\Lambda_k\}_{k = 1}^N$ in the following sense: if $\alpha(\bb{A}(\Lambda_k)) \subset_{\varepsilon_k} \bb{A}(\Sigma)$ for some collection $\varepsilon_k > 0$ and $\Sigma \subseteq \Z^d$, then $\bb{A}(\Sigma^c) \subseteq_{\varepsilon'} \alpha\left(\bb{A}\left(\cap \Lambda_k^c\right)\right)$, where $\varepsilon' = 2(\varepsilon_1 + \ldots + \varepsilon_n)$. 
To show this, define inductively the elements $B_{1, \ldots, k} := \alpha \circ \Pi_{\Lambda_k} \circ \alpha^{-1}(B_{1, \ldots, k-1})$, starting with $B_1 := \alpha \circ \Pi_{\Lambda_1} \circ \alpha^{-1}(B)$. Following the argument used for 
\eqref{eq:LLa}, one gets that 
$$\left\| \alpha^{-1}(B) - \Pi_{\Lambda_j} \circ \alpha^{-1}(B)\right\| \;\leqslant\; 2 \varepsilon_j \|B\|$$
for each $j \in \{1, \ldots, n\}$. Now, we get the following estimation:
\begin{align*}
\|B - B_{1, \ldots, n}\| \; &= \; \left\|\alpha^{-1}(B) - \Pi_{\Lambda_n} \circ \ldots \circ \Pi_{\Lambda_1}\left(\alpha^{-1}(B)\right)\right\|\\
&\leqslant\; \left\|\alpha^{-1}(B) - \Pi_{\Lambda_n}\left(\alpha^{-1}(B)\right)\right\| + \left\|\alpha^{-1}(B) - \Pi_{\Lambda_{n-1}} \circ \ldots \circ \Pi_{\Lambda_1}\left(\alpha^{-1}(B)\right)\right\|\\
&\leqslant\; \sum_{k = 1}^n \left\|\alpha^{-1}(B) - \Pi_{\Lambda_k}\left(\alpha^{-1}(B)\right)\right\|\\
&\leqslant\; 2\left(\varepsilon_1 + \ldots + \varepsilon_n\right)\|B\|,
\end{align*}
where the first inequality follows from $\|\Pi_{\Lambda_k}\| \leqslant 1$, and the second inequality is obtained by applying the argument inductively. Since $B_{1, \ldots, n}$ belongs to $\alpha\left(\bb{A}\left(\cap \Lambda_j^c\right)\right)$, we get the desired result.
Furthermore, the same argument extends to a countable collection of finite subsets. Indeed, let $\{\Lambda_k\}_{k \in \mathbb{N}}$ be such that $\alpha(\bb{A}(\Lambda_k)) \subset_{\varepsilon_k} \bb{A}(\Sigma)$ for all $k$, with $\varepsilon_k > 0$ and $\sum_k \varepsilon_k < \infty$. 
Then, $ \bb{A}(\Sigma^c) \subseteq_{2\sum_k \varepsilon_k} \alpha (\bb{A} (\cap \Lambda_k^c )).$ To see this, set $\widetilde{\Lambda}_n := \bigcup_{k = 1}^n \Lambda_k$ and $C_n = \alpha \circ \Pi_{\widetilde{\Lambda}_n} \circ \alpha^{-1}(B)$. Observe that the sequence $\{C_n\}$ converges to $C = \alpha \circ \Pi_\Lambda \circ \alpha^{-1}(B)$ where $\Lambda = \bigcup_k \Lambda_k$. Hence,
\begin{align*}
	||B - C|| &\;\leqslant\; \liminf_{n} ||B - C_n|| \;\leqslant\; 2 \liminf_n \sum_{k = 1}^n \varepsilon_k \; ||B||\; =\; \sum_{k = 1}^\infty \varepsilon_k \; ||B|| 
\end{align*}
as desired.
\hfill $\blacktriangleleft$
\end{remark}
\medskip

The next concept, based on \cite[Definition 3.5]{Ranarad-Walter-Witteveen-Converse-Lieb}, is needed to introduce an alternative characterization of the notion of an almost locality preserving automorphism.

\begin{definition}[Approximately locality preserving automorphisms]\label{ALP Definition 2}
    Let $\alpha \in {\rm Aut}(\bb{A})$ be an automorphism and $f \in \s{S}^+(\N_0)$. We say that $\alpha$ is an \emph{Approximately Locality-Preserving automorphism} with \emph{$f(r)$-tails} if for any ball $B_p(s)$, $p \in \Z^d$ and $s \in \N$, one has that $\alpha(\bb{A}(B_p(s)))$ is $s^{d-1} f(r)$-contained in $\bb{A}(B_p(s + r))$.
\end{definition}

For $r=0$ the condition in the definition is trivially satisfied whenever $f(0)\geqslant 1$. 
Although the Definition \ref{ALP Definition 2} for approximately locality preserving automorphisms is stated using the notion of near inclusions, in the next Proposition we show that the set of approximately locality preserving automorphisms coincides exactly with the set ${\rm ALP\text{-}Aut}(\bb{A})$. In this sense Definition \ref{ALP Definition 2} provides just a different characterization of the ALP-condition in Definition  \ref{ALP automorphism firts definition}. 

\begin{proposition}\label{Equivalence Definitions ALP}
    Definitions \ref{ALP automorphism firts definition} and \ref{ALP Definition 2} are equivalent.
\end{proposition}

\proof Let $\alpha: \bb{A} \rightarrow \bb{A}$ be an almost locality-preserving automorphism, and let $f^{(\alpha)}$ be the function from Definition \ref{ALP automorphism firts definition}. To show that $\alpha$ satisfies Definition \ref{ALP Definition 2}, we prove that for any ball $B_p(s)$, one has:
      $$\alpha(\bb{A}(B_p(s)))\; \subset_{s^{d-1}f^{(\alpha)}(r)}\; \bb{A}(B_p(s + r))\;.$$
    Indeed, for any $A \in \bb{A}(B_p(s))$, the definition of $f^{(\alpha)}$ implies
    $$\inf_{B \in \bb{A}(B_p(s + r))} \left|\left|\frac{1}{||A||} \alpha(A) - B\right|\right| \;\leq\; s^{d-1} f^{(\alpha)}(r)\;.$$
    Since $\bb{A}(B_p(s + r))$ is finite dimensional, the infimum is attained. This verifies that $\alpha$ is approximately locality-preserving with $f^{(\alpha)}-$tails. Conversely, suppose that $\alpha$ satisfies Definition \ref{ALP Definition 2} with $f(r)$-tails. Then, for any ball $B_p(s)$ and $A \in \bb{A}(B_p(s))$, we have
        $$\inf_{B \in \bb{A}(B_p(s + r))} ||\alpha(A) - B||\; \leq \;s^{d-1} f(r) ||A||\;.$$
Since this bound holds for any ball $B_p(s)$, taking the supremum as in Definition \ref{ALP automorphism firts definition}, one obtains $f^{(\alpha)}(r) = f(r) \in \s{S}^+(\N_0)$. This confirms that $\alpha$ is almost locality preserving.
 \qed

\medskip
 
The equivalent characterization of ALP-automorphisms in terms of the approximately locality preserving condition is a crucial ingredient to prove that set of ALP-automor\-phisms forms a group under composition. In particular, this result implies that the inverse of an ALP-automorphism $\alpha$ also belongs to this class. Before establishing that ALP-automorphisms indeed form a group, we present two technical results. The first is inspired by ideas from \cite[Lemma 3.4]{Ranarad-Walter-Witteveen-Converse-Lieb}, and the second by \cite[Lemma 3.3]{Ranarad-Walter-Witteveen-Converse-Lieb}.
In order to lighten the notation in the following,  let us introduce the integers
\[
z_{i,j}\;:=\;\left\lfloor\frac{1}{2}d(i,j)\right\rfloor
\]
where $\lfloor x \rfloor$ denotes the integer part of $x$, and $d(i,j)$ is the distance between $i$ and $j$. 

\begin{lemma}\label{ALP-Bound for regions}
	Let $\alpha\in {\rm ALP\text{-}Aut}(\bb{A})$ with $f(r)-$tails. Then, for any finite subset $\Sigma \subset \Z^d$ and any subset $\Lambda \subset \Z^d$, the following bound holds:
	\begin{equation}\label{eq0:ALP-Bound for regions}
	\bb{A}(\Sigma)\; \subseteq_{f_{\Lambda,\Sigma}
	} \;\alpha(\bb{A}(\Lambda^c))\;,
	\end{equation}
	       where
	    \begin{equation}\label{eq:F_LS}
	    f_{\Lambda,\Sigma}\;:=\;8 \sum_{i \in \Lambda} \sum_{j \in \Sigma} f(z_{i,j})\;,
	    \end{equation}
	and the convergence of the sum (for $\Lambda$ infinite) is guaranteed by the properties \eqref{VolumenBound} and \eqref{BoundaryBound} of the metric, along with the rapid decay of $f$.
\end{lemma}
\begin{proof}
	Let $\alpha$ be an ALP-automorphism with tail function $f(r)$, and assume without loss of generality that $f(0) = 1$. Let $i, j \in \Z^d$ and observe that $B_i(d(i,j)) \subset \{j\}^c$. If $d(i,j) < 2$, then we trivially have the inclusion
	\[
		\alpha(\bb{A}(\{i\}))\; \subseteq_{f(z_{i,j})} \;\bb{A}(\{j\}^c).
	\]
	Now, if $d(i,j) \geqslant 2$, we obtain, from the ALP condition, the chain of inclusions
	$$\alpha\left(\bb{A}(\{i\})\right) \; = \; \alpha\left(\bb{A}\left(B_i(1)\right)\right)\; \subset_{f(z_{i,j})}\; \bb{A}\left(B_{i}(1 + z_{i,j})\right)\; \subset\; \bb{A}(\{j\}^c)\;.$$  
According to Remark \ref{Generalization Near-Inclusions}, given a subset $\Lambda \subset \Z^d$, it follows that
$$\bb{A}(\{j\})\; \subset_{2 \sum_{i \in \Lambda} f(z_{i,j})}\; \alpha(\bb{A}(\Lambda^c))\;.$$
The convergence of the sum is guaranteed by the properties \eqref{VolumenBound} and \eqref{BoundaryBound} of the metric, along with the rapid decay of $f$. This relation can be equivalently written  as
$$\alpha^{-1}\left(\bb{A}(\{j\})\right)\; \subset_{2 \sum_{i \in \Lambda} f(z_{i,j})}\; \bb{A}(\Lambda^c)\;.$$
By considering again Remark \ref{Generalization Near-Inclusions} for $\alpha^{-1}$, one gets that
$$\bb{A}(\Lambda) \; \subseteq_{4 \sum_{i \in \Lambda} \sum_{j \in \Sigma} f(z_{i,j})} \; \alpha^{-1}(\bb{A}(\Sigma^c)).$$
Then, according to Lemma \ref{lemma_AX} one gets
$$\|[C', D']\|\; \leqslant\; \left(8 \sum_{i \in \Lambda} \sum_{j \in \Sigma} f(z_{i,j})\right) \|C'\| \; \|D'\|$$
for all $C' \in \bb{A}(\Lambda)$ and $D' \in \alpha^{-1}\left(\bb{A}(\Sigma)\right)$. This latter relation can be equivalently written as
$$\|[C, D]\|\; \leqslant\;  f_{\Lambda,\Sigma} \ \|C\| \; \|D\|\;,$$
with $C \in \alpha\left(\bb{A}(\Lambda)\right)$, $D \in \bb{A}(\Sigma)$ and the constant given by \eqref{eq:F_LS}. By applying Lemma \ref{Near-Inclusions} (1), one gets that 
$$
\bb{A}(\Sigma) \;\subseteq_{ f_{\Lambda,\Sigma}} \;\alpha\left(\bb{A}(\Lambda^c)\right),$$
which concludes the proof.
\end{proof}
\begin{lemma}\label{lemma: bounding sum regions}
	Let $F: \mathbb{R}_{\geqslant 0} \rightarrow \mathbb{R}_{\geqslant 0}$ be a monotonically decreasing function. Define the function $g_F: \N_0 \rightarrow \mathbb{R}_{\geqslant 0}$ by
	\begin{equation}\label{eq0: lemma bounding sum regions}
		g_F(t) \; := \; \sum_{k,l \geqslant 1}(t + k + l - 1)^{d - 1} F(t + k + l - 2)\;.
	\end{equation}
	Suppose that $g_F(t) < \infty$ for all $t \in \N_0$. 
	Then, there exists a constant $L_d > 0$ such that for any $p \in \Z^d$, $s \geqslant 1$, and $r \geqslant 0$, 
	\begin{equation}\label{eq2: lemma bounding sum regions}
		\sum_{i \in B_p(s)} \sum_{j \in B_p(s + r)^c} F(d(i,j)) \;\leqslant\; L_d \, s^{d-1} g_F(r)\;.
	\end{equation} 
\end{lemma}
\begin{proof}
	Given $p \in \Z^d$, $s \geqslant 1$ and $r \geqslant 0$, we aim to bound the quantity
	$$S_p(s,r) \; := \; \sum_{i \in B_p(s)} \sum_{j \in B_p(s + r)^c} F(d(i,j))\;.$$ 
	From the definition of $B_p(s)$ one gets that $B_p(s + r)^c = \{j \in \Z^d: \, d(p,j) > s + r - 1\}$. Let us consider the shells $\partial B_p(R + 1) = B_p(R + 1) \setminus B_p(R)$ for all $R \geqslant 0$. We can foliate the ball $B_p(s)$ using these disjoint shells
	$$B_p(s) = \bigcup_{R = 0}^{s-1} \partial B_p(R + 1)\;.$$
	Note that $\partial B_p(1) = \{p\}$. We re-index this foliation by letting $R = s - k$. The outer shell of $B_p(s)$ corresponds to $R = s-1$ which means $k = 1$. The inner-most shell, for which $R = 0$, corresponds to $k = s$. Let $\Lambda_k = \partial B_p(s - k + 1)$, for $k = 1, \ldots, s$. The foliation thus takes the form
	$$B_p(s) = \bigcup_{k = 1}^{s} \Lambda_k\;.$$
	From~\eqref{BoundaryBound}, it follows that $|\Lambda_k| \leqslant K_d (s - k)^{d-1}$ for $k < s$, and $|\Lambda_s| = 1$. Now, for each $i \in \Lambda_k$, we define the quantity
	$$M(i) \; := \; \sum_{j \in B_p(s + r)^c} F(d(i,j))\;.$$
	For any $j \in B_p(s + r)^c$, we can bound its distance to $i$ as
	$$d(j,i)\; \geqslant \;d(j,p) - d(i,p)\; > \;(r + s - 1) - (s - k)\; =\; r + k - 1\;.$$
	The first step uses the reverse triangle inequality, and the second uses the bounds $d(j,p) > s+r-1$ and $d(i,p) \le s-k$ (since $i \in \Lambda_k = \partial B_p(s-k+1)$).	
 The inequality $d(j,i) > r+k-1$ implies that $B_p(s + r)^c \subseteq B_i(r + k)^c$. Since $F$ is non-negative, we can bound $M(i)$ by summing over this larger, $i$-centered set
\begin{equation}\label{eq0: proof bounding sum regions}
	M(i)\; \leqslant \;\sum_{j \in B_i(r + k)^c} F(d(i,j))\;.
\end{equation}
Similarly to how we foliated $B_p(s)$, we can foliate the external domain $B_i(r + k)^c$ as
$$B_i(r + k)^c \; = \; \bigcup_{l = 1}^\infty \partial B_i (r + k + l)\;.$$
Given $j \in \partial B_i (r + k + l)$, we have $d(i,j) > r+k+l-2$. Since $F$ is monotonically decreasing, $F(d(i,j)) \leqslant F(r+k+l-2)$. The right hand side in~\eqref{eq0: proof bounding sum regions} is bounded as
\begin{equation}\label{eq1: proof bounding sum regions}
	\sum_{j \in B_i(r + k)^c} F(d(i,j)) \leqslant K_d \sum_{l = 1}^\infty (r + k + l - 1)^{d-1} F(r + k + l - 2),
\end{equation}
where $K_d$ is the geometric constant in~\eqref{BoundaryBound}. 
We can now bound  $S_p(s,r)$ by
\begin{align*}
	S_p(s,r) \;&=\; \sum_{i \in B_p(s)} M(i) \;=\; \sum_{k = 1}^{s} \sum_{i \in \Lambda_k} M(i) \\
	&\leqslant\; K_d \sum_{k = 1}^{s} |\Lambda_k| \left(\sum_{l = 1}^\infty (r + k + l - 1)^{d-1} F(r + k + l - 2)\right) \\
	&\leqslant\; K_d^2 \sum_{k = 1}^{s - 1}(s - k)^{d-1} \left(\sum_{l = 1}^\infty (r + k + l - 1)^{d-1} F(r + k + l - 2)\right) \\
	&\quad + K_d \sum_{l = 1}^\infty (r + s + l - 1)^{d-1} F(r + s + l - 2)\\
	&\leqslant\; K_d^2 s^{d-1} \sum_{k,l = 1}^{\infty} (r + k + l - 1)^{d-1} F(r + k + l - 2)\\
&=\; K_d^2 s^{d-1} g_F(r)\;.
\end{align*}
Note that this bound is uniform in $p \in \Z^d$, thus concluding the proof. Hence, the constant $L_d$ in the statement is equal to $K_d^2$.
\end{proof}
Having established these preliminary results, we now turn to prove the group property of the space of ALP-automorphisms.
\begin{proposition}\label{prop_grup}
The space of ${\rm ALP\text{-}Aut}(\bb{A})$ is a group under composition.
\end{proposition}
\begin{proof}
We first prove that ${\rm ALP\text{-}Aut}(\bb{A})$ is closed under composition.
Let $\alpha, \beta$ be two ALP-automorphisms with tail functions $f^{(\alpha)}$ and $f^{(\beta)}$, respectively. We will show that their composition $\alpha \circ \beta$ is also an ALP-automorphism. In fact, for any ball $B_p(s) \subset \Z^d$ and any $A \in \bb{A}(B_p(s))$ with $\|A\| = 1$, we will provide the bound
$$\inf_{B \in \bb{A}(B_p(s + r))} \|\alpha \circ \beta(A) - B\|\;\leq\; s^{d-1} h(r)$$
for some function $h\in\s{S}^+(\N_0)$. 
Define $n(r) := \lfloor r/2 \rfloor$ and $l(r) := r - \lfloor r/2 \rfloor$. Since $\bb{A}(B_p(R))$ is a finite-dimensional closed subalgebra, the infimum is attained for any $R$. 
Note that $B_p(s + l(r))\subseteq B_p(s + r)$. Let $D \in \bb{A}(B_p(s + l(r)))$ be such that
\begin{equation*}
	\inf_{B \in \bb{A}(B_p(s + l(r)))} \|\beta(A) - B \| \;=\; \|\beta(A) - D\|\;.
\end{equation*} 
By the ALP condition for $\beta$, we have
	\begin{equation}\label{eq0: gp}
		\|\beta(A) - D\| \;\leqslant \;s^{d-1} f^{(\beta)}(l(r))\;.
	\end{equation}
	Moreover, using the reverse triangle, one gets
\begin{equation}\label{eq1: gp}
	\|D \|\; \leq\; \inf_{B \in \bb{A}(B_p(s + l(r)))} \|\beta(A) - B \|+\|\beta(A)\|
	\; \leq\;
	2 \;.
\end{equation}
Similarly, let $C\in \bb{A}(B_p(s + r))$ be the element attaining the infimum for $\alpha(D)$:
$$
\inf_{B \in\bb{A}(B_p(s + r))} \|\alpha (D) - B \| \;=\; \|\alpha(D) - C\|\;.
$$
Since $D$ is supported in $B_p(s + l(r))$ and $s + l(r) + n(r) = s + r$, the ALP condition for $\alpha$ implies
\begin{equation}\label{eq2: gp}
\|\alpha(D) - C\| \leqslant (s + l(r))^{d-1} \, f^{(\alpha)}(n(r))  \; ||D|| \leqslant  2^{d-1} (s^{d-1} + l(r)^{d-1}) \, f^{(\alpha)}(n(r))  \; ||D||\;,
\end{equation}
where, in the last bound, we have used the inequality $(a + b)^n \leqslant 2^n (a^n + b^n)$ for any $n \in \N_0$ and $a,b > 0$.
 Combining~\eqref{eq0: gp}, \eqref{eq1: gp}, and~\eqref{eq2: gp}, one obtains
\begin{align*}
	\inf_{B \in \bb{A}(B_p(s + r))} \|\alpha ( \beta(A)) - B\| \;& \leq\; \|\alpha( \beta (A)) - C - \alpha(D) + \alpha(D)\| \\
	& \leq\; \|\beta(A) - D\| + \|\alpha(D) - C\| \\
	& \leq \;s^{d - 1} f^{(\beta)}(l(r)) + 2^{d-1} (s^{d-1} + l(r)^{d-1}) f^{(\alpha)}(n(r)) \|D\| \\
	& \leq \; s^{d-1} \left( f^{(\beta)}(l(r)) + 2^{d} (1 + l(r)^{d-1})f^{(\alpha)}(n(r)) \right)\;.
\end{align*}
Therefore, one concludes that $\alpha \circ \beta$ has a tail function provided by
\begin{equation}\label{eq3: gp}
	h(r)\; :=\; f^{(\beta)}(l(r)) + 2^{d} (1 + l(r)^{d-1})f^{(\alpha)}(n(r))\;.
\end{equation}
Now we prove that ${\rm ALP\text{-}Aut}(\bb{A})$ is closed under inversion. Let $\alpha$ be an ALP automorphism with tail function $f^{(\alpha)}$. To show that $\alpha^{-1}$ is also an ALP-automorphism, we will use Lemma \ref{ALP-Bound for regions}. Specifically, we take $\Sigma = B_p(s)$ and $\Lambda = B_p(s + r)^c$ in the statement of the lemma. Then, by applying $\alpha^{-1}$ to both sides of the inclusion \eqref{eq0:ALP-Bound for regions}, we obtain:
$$\alpha^{-1}(\bb{A}(B_p(s)))\; \subseteq_{\widetilde{h}_p(s, r)} \;\bb{A}(B_p(s + r))$$
where to simplify the notation we introduced
the function $\widetilde{h}_p: \N \times \N_0 \rightarrow \R_{+}$ given by
$$\widetilde{h}_p(s, r) \;:= \;f^{(\alpha)}_{B_p(s + r)^c, B_p(s)} = 8 \sum_{i \in B_p(s + r)^c} \sum_{j \in B_p(s)} f^{(\alpha)}(z_{i,j}) $$
for each $p \in \Z^d$. To prove that $\alpha^{-1}$ is ALP, it suffices to show that there exists ${h} \in \s{S}(\N_0)$ such that $\widetilde{h}_p(s,r) \leqslant s^{d-1} h(r)$ uniformly in $p$. Since $f^{(\alpha)}$ belongs to $\s{S}^+(\N_0)$, this function satisfies the hypothesis in Lemma~\ref{lemma: bounding sum regions}. Therefore, for any $p \in \Z^d,$ $s \geqslant 1$ and $r \geqslant 0$, it follows that
$$\widetilde{h}_p(s, r)\;\leqslant\; 8L_d s^{d-1} \sum_{k,l \geqslant 1} (r + k + l - 1)^{d-1} f^{(\alpha)}\left(\left\lfloor \frac{r + k + l - 2}{2}\right\rfloor\right)\;.$$
Finally, since $f^{(\alpha)} \in \s{S}^+(\N_0)$, it then follows that the function 
\begin{equation}\label{eq4: gp}
	h(r) \; := \;
	8L_d \sum_{k,l \geqslant 1} (r + k + l - 1)^{d-1} f^{(\alpha)}\left(\left\lfloor \frac{r + k + l - 2}{2}\right\rfloor\right)
\end{equation}  
belongs to $\s{S}(\N_0)$. Here, by the discussion in Appendix \ref{app:Rap_seq}, we can assume that $h$ in fact belongs to $\s{S}^+(\N_0)$.
\end{proof}
\begin{remark}
In the references~\cite{Kapustin-Sopenko-Lieb} and~\cite{Ranarad-Walter-Witteveen-Converse-Lieb} the authors consider the notion of almost locality-preserving (ALP) automorphisms in the one-dimensional case with the standard metric on $\Z$. In \cite{Kapustin-Sopenko-Lieb}, the definition they adopt is similar to Definition \ref{ALP automorphism firts definition}, but instead of balls, the condition is required to hold over all finite intervals in $\Z$. On the other hand, in \cite{Ranarad-Walter-Witteveen-Converse-Lieb}, the authors employ a version similar to Definition \ref{ALP Definition 2}, but the approximation is required to hold uniformly over all intervals --finite or infinite. In this one-dimensional setting, they show that if $\alpha$ is ALP with tail function $f$, then $\alpha^{-1}$ is also ALP with tail function $4f$, which follows from their more rigid interval-based setup.
In our work, motivated by the goal of extending the notion of almost locality-preserving automorphisms to arbitrary spatial dimension $d$, we have slightly adapted the definition and focused instead on metric balls as basic regions. While this adjustment sacrifices some of the simplicity, such as the neat $4f$ bound for $\alpha^{-1}$, it allows us to formulate a consistent definition valid for general lattices $(\Z^d, d)$, beyond the $d=1$ case.
\hfill $\blacktriangleleft$
\end{remark}
We are now in position to justify a fact left open at the end of the Section \ref{sec-ALP}.
\begin{corollary}\label{corol_aut-F}
Any element $\alpha\in {\rm ALP\text{-}Aut}(\bb{A})$ restricts to an automorphism of the Fr\'echet algebra $\bb{A}_\infty$.
\end{corollary}
\proof
Proposition \ref{prop_grup} implies that also $\alpha^{-1}\in {\rm ALP\text{-}Aut}(\bb{A})$. By Propositions \ref{Proposicion Continuidad ALP en Frechet}, one concludes that both $\alpha$ and $\alpha^{-1}$
restricts to continuous $\ast$-homomorphisms of the Fr\'echet algebra $\bb{A}_\infty$. Therefore, such a restriction is a automorphism of $\bb{A}_\infty$.
\qed

\subsection{Lieb-Robinson type automorphisms}
With the Fr\'echat algebra and its automorphism group in hand, we return to Lieb-Robinson bounds. We will follows here the definition of \cite{Nachtergaele-Sims-Young-Quasilocality-bounds} and \cite{Naaijkens-peter}.

\medskip

One says that $F:[0,+\infty]\to[0,+\infty)$ is a \emph{reproducing function} (or $F$-function) for $(\Z^d,d)$ if:
    \begin{enumerate}
        \item[(i)] it is uniformly integrable, meaning that
        $$||F|| \;:= \; \sup_{x \in \Z^d} \sum_{y \in \Z^d} F(d(x,y))\; <\; +\infty\;;$$
        \vspace{1mm}
        \item[(ii)] it satisfies a convolution condition
        $$
        C_F\; :=\; \sup_{x, y \in \Z^d} \sum_{z \in \Z^d} \frac{F(d(x,z))F(d(z,y))}{F(d(x,y))}\; <\; +\infty\;.
        $$
    \end{enumerate}

\medskip

While the following definition does not refer to any derivation, it is inspired by results on well-behaved derivations in spin algebras, see in particular \cite[Corollary 3.6]{Nachtergaele-Sims-Young-Quasilocality-bounds} and~\cite[Lemma 3.3]{Ranarad-Walter-Witteveen-Converse-Lieb}.
\begin{definition}[Lieb-Robinson type automorphisms]\label{defLR-aut}
    Let $\alpha\in {\rm Aut}(\bb{A})$ be an automorphism on the quasi-local algebra $\bb{A}$. We say that $\alpha$ is of \emph{Lieb-Robinson type} if there exists a monotonically decreasing  reproducing function $F$  such that:
    \begin{enumerate}
        \item[(i)] the inequality 
        $$||[\alpha(A), B]||\; \leq\; ||A|| \ ||B|| \sum_{i \in \Lambda} \sum_{j \in \Sigma} F(d(i,j))$$
        holds for all finite subsets $\Lambda,\Sigma \in \s{P}_0(\Z^d)$ and for all $A \in \bb{A}(\Lambda)$, $B \in \bb{A}(\Sigma)$; and
                \vspace{1mm}
        \item[(ii)] the function $f_F: \N_0 \to [0,+\infty)$, defined by
        \begin{equation}\label{eq:f_F}
        f_F(r)\;: =\; K_d^2 g_F(r)\;,
       \end{equation}
         where $g_F$ is the function defined in \eqref{eq0: lemma bounding sum regions}, belongs to $\s{S}(\N_0)$, and the constant $K_d$ is that appearing in the metric condition \eqref{BoundaryBound}.
    \end{enumerate}
\end{definition}

\medskip

\begin{remark}
	For reproducing functions associated with the lattice $(\Z^d,d)$, one of the most natural and widely used examples are those with polynomial decay, such as
	\begin{equation}\label{eq: Repr Function Pol}
		F_{\rm Pol}(r) \;=\; \frac{1}{(1+r)^{d - 1 + \nu}}, \qquad \nu > 1.	
	\end{equation}
A detailed discussion of this family of reproducing functions can be found in \cite[Appe\-ndix~A]{Nachtergaele-Sims-Young-Quasilocality-bounds}. We shall return to the class of automorphisms associated with such functions in the next section. 
For the present rapidly decaying case, however, introducing an exponential weight
\[
F_a(r) \;=\; e^{-a r}\,F_{\rm Pol}(r), \qquad a>0,
\]
yields a reproducing function whose associated $f_{F_a}$ satisfies conditions (i) and (ii) of Defi\-nition~\ref{defLR-aut}.
In the next section, we shall return to automorphisms satisfying condition~(i) but associa\-ted with reproducing functions of the form~\eqref{eq: Repr Function Pol}, see Definition~\ref{Polynomial LR def}.
\hfill $\blacktriangleleft$
\end{remark}

\begin{remark}
In physical application one usually start with an injective $\ast$-homomorphism
$\alpha':\bb{A}_{\rm loc}\to \bb{A}$ defined initially only on local observable which verifies properties (i) and (ii) in Definition \ref{defLR-aut}. Then  $\alpha'$ extends to a unique automorphism $\alpha\in {\rm Aut}(\bb{A})$ of Lieb-Robinson type by continuity. Since $\alpha$ restricted to any subalgebra $\bb{A}(\Lambda)$ for $\Lambda \subset \Z$ finite coincides with $\alpha'$, it follows trivially that $\alpha$ fulfils condition (i) in Definition \ref{defLR-aut}.
\hfill $\blacktriangleleft$
\end{remark}

\medskip

In the following proposition, we show that any automorphism $\alpha$ of Lieb-Robinson type in the sense of Definition~\ref{defLR-aut} is almost locality-preserving in the sense of Definition~\ref{ALP Definition 2}. 
In turn, any Lieb-Robinson type automorphism is Fr\'echet continuous. 

\begin{proposition}\label{Lemma Nuevo Lieb Robinson}
    Let $\alpha\in {\rm Aut}(\bb{A})$ be an automorphism on the spin algebra $\bb{A}$. Suppose there exists a monotonically decreasing reproducing function $F:[0,+\infty)\to[0,+\infty)$ such that     
    \[
    \|[\alpha(A), B]\|\; \leq \;\|A\| \; \|B\| \sum_{i \in \Lambda_1} \sum_{j \in \Lambda_2} F(d(i,j))\;.
    \]
for all $\Lambda_1, \Lambda_2 \subseteq \Z^d$ and all $A \in \bb{A}(\Lambda_1)$, $B \in \bb{A}(\Lambda_2)$. Assume further that for all $t \geqslant 0$, the quantity 
$$\sum_{k,l = 1}^\infty (t + k + l - 1)^{d - 1} F(t + k + l - 2)$$
 is finite. Then, for any ball $B_p(s)$, one has
\[
\alpha(\bb{A}(B_p(s)))\; \subset_{s^{d-1}f_F(r)}\; \bb{A}(B_p(s + r))\;,
\]
    where $f_F$ is the function defined in \eqref{eq:f_F}. In particular, if $F$ is rapidly decreasing, the automorphism $\alpha$ is an ALP automorphism in the sense of Definition~\ref{ALP Definition 2}.
    \end{proposition}
\proof Let $p \in \Z^d$, $s \geqslant 1$ and $r \geqslant 0$. According to Lemma \ref{Near-Inclusions} (1), it is enough to prove that
 \[
 \|[\alpha(A), B]\|\; \leq\; s^{d-1} f(r) \|A \| \; \|B\|\;,
 	\]
for any $A \in \bb{A}(B_p(s))$, any $B \in \bb{A}(B_p(s + r)^c)$ and an appropriately defined $f(r)$. By hypothesis on $F$, one gets from Lemma~\ref{lemma: bounding sum regions} that the function $f_F$ in~\eqref{eq:f_F} satisfies
$$\sum_{i \in B_p(s)} \sum_{j \in B_p(r + s)^c} F(d(i,j)) \;\leqslant\; s^{d-1} f_F(r)\;.$$
Hence, by hypothesis on $\alpha$ and $F$, one obtains the desired result. For the special case with $F$ rapidly decreasing, it follows that $f_F$ is also rapidly decreasing, and thus $\alpha$ is almost locality preserving with tail $f_F(r)$.
\qed

\medskip

The next step is to prove the converse of Proposition \ref{Lemma Nuevo Lieb Robinson}. To do so, we will need a technical result whose proof is provided in \cite[Lemma E.2]{Kapustin-Sopenko-Local-Noether}. 
\begin{lemma}\label{Bounding by reproducing function}
    Let $f \in \s{S}^+(\N_0)$. Then there exists a reproducing function $F$ for $(\Z^d,d)$, such that $F \in \s{S}^+(\N_0)$, and $f(r) \leq F(r)$ for any $r \in \N_0$. 
\end{lemma}
\begin{proposition}\label{ALP implies LB}
    Let $\alpha\in {\rm ALP\text{-}Aut}(\bb{A})$. Then $\alpha$ is of Lieb-Robinson type.
\end{proposition}
\begin{proof}
    Let $\alpha$ be an ALP automorphism with tail function $f(r)$. Without loss of generality, let us assume that $f(0) = 1$. Let $\Sigma, \Lambda \subset \Z^d$ be finite subsets. From the proof of Lemma \ref{ALP-Bound for regions}, one has the bound
$$\| [\alpha(A), B] \| \leqslant 8 \sum_{i \in \Lambda} \sum_{j \in \Sigma} f(z_{i,j}) \|A\| \, \|B\|,$$
for all $A \in \bb{A}(\Lambda)$ and $B \in \bb{A}(\Sigma)$.
Now, according to Lemma \ref{Bounding by reproducing function}, there exists a reproducing function $F \in \s{S}^+(\N_0)$ for $(\Z^d,d)$ such that $8 \; f(r) \leqslant F(r)$. Hence, one gets (i) in Definition \ref{defLR-aut}.	Since the resulting reproducing function $F$ belongs to $\s{S}^+(\N_0)$, it follows that the function $f_F$ in condition (ii) in Definition \ref{defLR-aut}, also belongs to $\s{S}(\N_0)$. Therefore, we conclude that $\alpha$ is of Lieb-Robinson type. 
\end{proof}

\subsection{The Polynomial case}\label{Pol Case sec}
In this section, we extend several of the previous results to a broader class of interactions, namely those exhibiting at most polynomial decay. To this end, we introduce a notion of elements with polynomial decay, defined in analogy with the Fr\'echet algebra $\bb{A}_\infty$, and we adapt the concept of almost locality-preserving automorphisms to this setting, allowing for polynomially decaying tails. We then define Lieb–Robinson type automorphisms associated with such polynomially decaying functions $F_{\rm Pol}$.

\medskip

As in the rapidly decaying case, ALP automorphisms whose tails decay as $r^{-k-d+1}$ are shown to be continuous $*$-homomorphisms on the algebra of $k$–decaying observables. in contrast to the rapidly decaying case, the techniques developed in this work do not guarantee that these homomorphisms form a group under composition. We further establish a link between polynomially decaying ALP automorphisms and Lieb-Robinson type automorphisms asso\-ciated with $F_{\rm Pol}$, highlighting a nontrivial role played by the lattice dimension $d$. Finally, we show that the family of flip automorphisms $\psi_\zeta$, with $\zeta$ any polynomial, remains continuous with respect to all these topologies, despite not satisfying the polynomial ALP condition.

\begin{definition}[$(k)$ polynomial local observables]\label{Def: Polynomial local observables}We define the algebra $\bb{A}_{(k)}$ of $(k)-$ polynomially decaying elements of $\bb{A}$ as the completion of $\bb{A}_{\rm loc}$ with respect to the norm $|||\cdot|||'_k$ (or, equivalently, with respect to $||\cdot||_k$). We refer to $\bb{A}_{(k)}$ as the algebra of $(k)-$decaying elements of $\bb{A}$. 
\end{definition}

\begin{proposition}\label{Properties K polynomial decaying}
The space $\bb{A}_{(k)}$ corresponds to those elements $A \in \bb{A}$ such that: (i) $||A||_k < \infty$, and (ii) $\lim_{r \to \infty} (1 + r)^k p_r(A) = 0$.
\end{proposition}

\begin{proof}
First of all, let us observe that, in fact, $\bb{A}_{(k)}$ is a subspace of $\bb{A}$. Let $\{A^{(n)}\}_{n \in \N} \subseteq \bb{A}_{\rm loc}$ be a Cauchy sequence with respect to the norm $\| \cdot \|_k$. Define the functions $g_n(r) := p_r(A^{(n)})$. By the definition of $||\cdot||_k$,
$$||A^{(n)} - A^{(m)}||_k \; = \; ||A^{(n)} - A^{(m)}|| + \sup_{r \in \N_0}(1 + r)^k p_r(A^{(n)} - A^{(m)}),$$
we deduce that the sequence $\{A^{(n)}\}_{n \in \N}$ is also Cauchy with respect to the operator norm. Therefore, it converges to some element $A \in \bb{A}$. Additionally, since $|p_r(A^{(n)}) - p_r(A^{(m)})| \leqslant p_r(A^{(n)} - A^{(m)})$, the sequence of functions $g_n(r)$ is Cauchy with res\-pect to the norm $|||\cdot |||_k$. Let $g(r) := p_r(A)$. By continuity of the partial traces, and since $A^{(n)} \to A$ in norm, it follows that $g_n(r)$ converges pointwise to $g(r)$. To prove convergence with respect to $|||\cdot|||_k$, fix $M \in \N$ and consider the following estimate:
\begin{align*}
\sup_{r \in [0, M]} (1 + r)^k |g_n(r) - g(r)| &\leqslant |||g_n(r) - g_m(r)|||_k + \sup_{r \in [0, M]} (1 + r)^k |g_m(r) - g(r)|.
\end{align*}
Given $\epsilon > 0$, choose $N_\epsilon$ such that for $n,m \geqslant N_\epsilon$, $|||g_n(r) - g_m(r)|||_k < \epsilon/2$. Then, fix $m \geqslant N_\epsilon$ so that by pointwise convergence, $|g_m(r) - g(r)| < \frac{\epsilon}{2(1 + M)^k}$ for all $r \leqslant M$. Hence
$$\sup_{r \in [0, M]} (1 + r)^k |g_n(r) - g(r)| < \epsilon.$$
Since the bound is uniform in $M$, it follows that $g_n \to g$ in the norm $|||\cdot|||_k$. \\
Moreover, since the sequence $g_n(r)$ has a limit in the $|||\cdot|||_k-$topology, there exists $C > 0$ such that $|||g_n|||_k \leqslant C$ for all $n$. By the triangle inequality, 
$$|||g|||_k \leqslant |||g - g_n|||_k + |||g_n||| \leqslant \epsilon + C,$$ so $|||p_r(A)|||_k <\infty$, and $||A||_k < \infty$. This proves that the closure of $\bb{A}_{\rm loc}$ is a subspace of the algebra $\bb{A}$ and that the elements of the closure satisfies $(i)$. Observe now that for any $r \in \N_0$ and fixed $n$, one gets
$$(1 + r)^k p_r(A) \leqslant ||A - A^{(n)}||_k + (1 + r)^k p_r(A^{(n)}).$$
Since $A^{(n)}$ is locally supported, there exists $R_n$ such that $p_r(A^{(n)}) = 0$ for $r \geqslant R_n$. Therefore, $\lim_{r \to \infty} (1 + r)^k p_r(A) = 0$. This shows that the elements in $\bb{A}_{(k)}$ satisfy $(i)$ and $(ii)$. \\
It remains to show the converse. Let $A \in \bb{A}$ such that $||A||_k < \infty$ and $\lim_{r \to \infty}(1 + r)^k p_r(A) = 0$. Define the locally supported sequence $A^{(n)} := \Pi_{B_0(n)^c}(A)$. Notice that $p_n(A) = ||A - A^{(n)}||$. Since for $r \geqslant s$ one has that $\Pi_{B_0(r)^c} \circ \Pi_{B_0(s)^c} = \Pi_{B_0(s)^c}$, one deduces that
$$p_r(A - A^{(n)}) \; = \; \left\{
        \begin{aligned}
        &||A - A^{(n)}||&\text{if}\;\;&0 \leqslant r \leqslant n - 1\\
        &p_r(A)&\text{if}\;\;&r \geqslant n
        \end{aligned}
        \right. .$$
Hence, it follows that
$$||A- A^{(n)}||_k \; = \; p_n(A) + \sup_{r \geqslant n} (1 + r)^k p_r(A).$$
By assumption (i), the term $p_n(A)$ tends to zero, and by (ii), the supremum term also vanishes as $n \to \infty$. Hence, $A$ belongs to $\bb{A}_{(k)}$.
\end{proof}

\begin{remark}
Analogously to Proposition \ref{Frechet algebra properties} for the Fr\'echet algebra $\bb{A}_\infty,$ the algebra $\bb{A}_{(k)}$ inherits a Banach $*-$algebra structure. Furthermore, one gets the chain of inclusions:
$$\bb{A} \supseteq \bb{A}_{(1)} \supseteq \bb{A}_{(2)} \supseteq \ldots \supseteq \bb{A}_\infty \supseteq \bb{A}_{\rm loc}.$$
\hfill $\blacktriangleleft$
\end{remark}

Following the same approach used for $\bb{A}_\infty$, we say that an $*-$homomorphism $\alpha: \bb{A}_{(k)} \rightarrow \bb{A}_{(k)}$ is continuous if there exists a constant $C_k > 0$ such that $|||\alpha(A)|||'_k \leqslant C_k |||A|||'_k$ for all $A \in \bb{A}_{(k)}$. We say that $\alpha$ is an automorphism if it and its inverse $\alpha^{-1}$ are continuous. The group of all such automorphisms of the Banach algebra $\bb{A}_{(k)}$ will be denoted as $\operatorname{Aut}(\bb{A}_{(k)})$. Using an argument analogous to that of Proposition \ref{Frechet continuity implies Operator continuity}, any automorphism $\alpha \in \operatorname{Aut}(\bb{A}_{(k)})$ extends uniquely to an automorphism $\tilde{\alpha}$ on the whole algebra $\bb{A}$. 

\medskip

We now introduce a weaker version of ALP-type automorphisms, characterized by polynomially decaying tails. To this end, we follow Definition~\ref{ALP automorphism firts definition}.

\begin{definition}\label{ALP-Polynomial Def}
	Let $\alpha \in \operatorname{Aut}(\bb{A})$ be an automorphism, and recall the function $f^{(\alpha)}$ from Definition~\ref{ALP automorphism firts definition}. We say that $\alpha$ is \emph{$(k)$-almost locality preserving} ($(k)$-ALP) if $f^{(\alpha)}$ decays at least as fast as a power law of order $k + d - 1$. That is, if the following conditions hold:
\[
|||f^{(\alpha)}|||_{k + d - 1} \; = \; \sup_{r \in \N_0}(1 + r)^{k + d - 1} f^{(\alpha)}(r) < \infty \quad \text{and} \quad \lim_{r \to \infty} (1 + r)^{k + d - 1} f^{(\alpha)}(r) = 0.
\]
We denote the space of $(k)$-ALP automorphisms as $(k)-\operatorname{ALP-Aut}(\bb{A})$.
\end{definition}
\begin{remark}
	The additional power factor of $d - 1$ in the decay rate can be traced back to the definition of $H_\alpha(s,r)$ (Definition~\ref{ALP automorphism firts definition}). This extra power is necessary to control the $s^{d-1}$ prefactor that appears in the general ALP bound. 
	\hfill $\blacktriangleleft$
\end{remark}
\begin{proposition}
	Any $(k)-$ALP-automorphism restricts to a continuous $*-$homomorphism on the Banach $*-$algebra $\bb{A}_{(k)}$. Furthermore, for $k \geqslant l$, one has 
	$$(k)-\operatorname{ALP-Aut}(\bb{A}) \subseteq (l)-\operatorname{ALP-Aut}(\bb{A}).$$
\end{proposition}

\begin{proof}
	The proof of Proposition \ref{Proposicion Continuidad ALP en Frechet} applies equally well in this setting. There, in expression~\eqref{eq;mm2} an upper bound for $p_r(\alpha(A))$ was established in terms of $D^{(\alpha)}(r)$ and $p_r(A)$. From this estimate, one concludes that $\alpha(\bb{A}_{(k)}) \subseteq \bb{A}_{(k)}$. Furthermore, expression~\eqref{eq;mm3} tells us that $\alpha$ is continuous with respect to the $||\cdot||_k-$topology. The second claim follows directly from the definition.
\end{proof}

\begin{remark}
	Once again, one can provide an alternative definition of the space of $(k)$-ALP automorphisms, following an approach analogous to that of Definition \ref{ALP Definition 2}. Specifi\-cally, we say that an automorphism $\alpha$ is an $(k)-$\textit{approximately Locality-Preserving automorphism} with $f(r)-$tails if there exists a non-increasing decaying function $f: \N_0 \rightarrow \R_{\geqslant 0}$ such that $|||f|||_{k + d - 1} < \infty$ and $\lim_{r \in \infty} r^{k + d - 1}f(r) = 0$, and such that for any ball $B_p(s) \subset \Z^d$ one has that $\alpha(\bb{A}(B_p(s)))$ is $s^{d-1} \, f(r)-$contained in $\bb{A}(B_p(s + r))$. 
Using the same argument as in Proposition \ref{Equivalence Definitions ALP}, one shows that this definition is equivalent to the one given in terms of the decay of $f^{(\alpha)}$.
	\hfill $\blacktriangleleft$
\end{remark}

\begin{remark}
	In contrast with the ``smooth'' case of ALP automorphisms, for the families $(k)-\operatorname{ALP-Aut}(\bb{A})$ our current estimates do not directly guarantee that they form a group under composition. In particular, the arguments used in the proof of Proposition~\ref{prop_grup} do not ensure that the tail associated with the inverse map $\alpha^{-1}$ still decays polynomially with order $k + d -1$.
	\hfill $\blacktriangleleft$
\end{remark}

Now, we turn to the class of Lieb–Robinson type automorphisms associated with the polynomially decaying reproducing function $F_{\mathrm{Pol}}$ defined in~\eqref{eq: Repr Function Pol}. These polynomial Lieb–Robinson type bounds are of particular relevance in physics, as they naturally arise in the study of long-range interactions, namely those whose strength decays polynomially with the distance \cite{Improved LR polynomial,LucasLongLRB}.

Recall the polynomially decaying reproducing function $F_{\mathrm{Pol}}$ as in~\eqref{eq: Repr Function Pol}. Using this choice of $F_{\rm Pol}$, one obtains the following estimate for the associated function $f_{F_{\rm Pol}}$ defined in~\eqref{eq:f_F}:
\begin{align*}
	f_{F_{\rm Pol}}(r) &\; = \; K_d^2 \sum_{k,l \geqslant 1} \frac{(r + k + l - 1)^{d -1}}{(r + k + l - 1)^{d - 1 + \nu}}
	\; = \; K_d^2 \sum_{k,l \geqslant 1} (r + k + l - 1)^{-\nu} \\[2mm]
	& \; = \; K_d^2 \sum_{m \geqslant 1} \frac{m}{(r + m)^\nu} \; = \; \mathcal{O}(1/r^{\nu - 2}),
\end{align*}
where the last estimate is obtained by bounding the sum with an integral. Note that this estimate is meaningful provided $\nu >  2$.

\medskip

We now reformulate Definition~\ref{defLR-aut} to include automorphisms whose interaction strength decays polynomially, thereby extending the notion of Lieb–Robinson type bounds to the polynomial setting.
 
\begin{definition}\label{Polynomial LR def}
Let $\alpha$ be an automorphism on the algebra $\bb{A}$, and let $F_{\rm Pol}:[0,\infty) \rightarrow [0, \infty)$ be a reproducing function for $(\Z^d,d)$ of the form \eqref{eq: Repr Function Pol} for $\nu \in \Z_{\geqslant d + 1}$. We say that $\alpha$ is a $(\nu)-$Lieb-Robinson type if the inequality 
        $$||[\alpha(A), B]||\; \leq\; C \ ||A|| \ ||B|| \sum_{i \in \Lambda} \sum_{j \in \Sigma} F_{\rm Pol}(d(i,j))$$
        for some constant $C > 0$, holds for all finite subsets $\Lambda,\Sigma \in \s{P}_0(\Z^d)$ and for all $A \in \bb{A}(\Lambda)$, $B \in \bb{A}(\Sigma)$.
\end{definition}

The stronger condition $\nu \geqslant d + 1$ is imposed to establish a direct connection between this class of Lieb-Robinson type dynamics and the family of $(k)$-$\operatorname{ALP}$ automorphisms introduced above. Moreover, although this may initially seem like a technical requirement, in \cite{Kuwahara} the authors argue that such polynomial decay---namely interactions of order $1/R^{\alpha}$ with $\alpha > 2d + 1$---is in fact necessary to guarantee the emergence of a linear light-cone.

\begin{proposition}\label{Prop: Pol LR ALP}
	Let $\alpha$ be a $(\nu)-$Lieb–Robinson type automorphism on $\bb{A}$ for $\nu \geqslant d + 1$. Then $\alpha$ is a $(\nu - d - 1)$–ALP automorphism.
\end{proposition}
\begin{proof}
	Let $\alpha$ be an automorphism of $(\nu)-$Lieb-Robinson type. Arguing analogously to Proposition~\ref{Lemma Nuevo Lieb Robinson}, and using that $f_{F_{\rm Pol}}(r) = \mathcal{O}\left(1/r^{\nu - 2}\right)$, we conclude that $\alpha$ is indeed a $(\nu - d - 1)-$ALP automorphism.
\end{proof}

\medskip

Finally, the question arises --as in the case of the Fr\'echet algebra $\bb{A}_\infty$--whether every automorphism of the Banach $*-$algebra $\bb{A}_{(k)}$ is a $(k)-$ALP automorphism. To explore this, we once again consider a flip-type automorphism as in Example~\ref{Example Not ALP}.

\begin{proposition}\label{Prop_Flip_kALP}
	For any polynomial $\zeta$ of degree at least 2, the corresponding flip automorphism $\psi_\zeta$ defines a continuous $*-$homomorphism $\psi_\zeta: \bb{A}_{(k)} \rightarrow \bb{A}_{(k)}$ for every $k \in \N$. However, none of these automorphisms belongs to any of the classes $(k)-\operatorname{ALP-Aut}(\bb{A})$.
\end{proposition}
	
\begin{proof}
	Recall from Example~\ref{Example Not ALP} that the function $H_{\psi_\zeta}$ associated with a flip automorphism satisfies $H_{\psi_\zeta}(s, r) \geqslant 1/2$ for any $s \geqslant 1$ and $r \geqslant 0$. Therefore, it follows that $\psi_\zeta$ does not belong to any class $(k)-\operatorname{ALP-Aut}(\bb{A})$.
	
	On the other hand, following the arguments of Subsection~\ref{Example Not ALP-2}, one has that for each $k \in \N$, there exists $M_k > 0$, depending only on $k$ and on the polynomial degree of $\zeta$, such that:
	$$||\psi_\zeta(A)||_k \; \leqslant \; M_k ||A||_k$$
	for all $A \in \bb{A}_{(k)}$. Now, to show that $\psi_\zeta(\bb{A}_{(k)}) \subseteq \bb{A}_{(k)}$, we need to verify that for any $A \in \bb{A}_{(k)}$, the following holds:
	$$\lim_{r \to \infty} (1 + r)^k p_r(\psi_\zeta(A)) = 0.$$
	In the proof of Fr\'echet continuity of $\psi_\zeta$, an increasing sequence $\{r_n\}_{n \in \N}\subset \N$--depending on $\zeta$--and a constant $C_\zeta > 0$ were constructed such that, for any $A \in \bb{A}_{\rm loc}$ and any $s \in [0, r_{n+1}-r_n-1]$,
$$
(1 + r_n + s)^k p_{r_n+s}(\psi_\zeta(A)) \;\leqslant\; 2 (C_\zeta + 1)^k (1 + r_n)^k\, p_{r_n}(A).
$$
	This inequality implies that $\lim_{r \to \infty} (1 + r)^k p_r(\psi_\zeta (A)) = 0$ for all $A \in \bb{D}_{\rm loc}$. Since $p_r(A - \omega_\infty(A)) = p_r(A)$ for all $r$, the same conclusion holds for all $A \in \bb{A}_{\rm loc}$. 
	Finally, for arbitrary $A \in \bb{A}_{(k)}$ and any $B \in \bb{A}_{\rm loc}$, one has:
\begin{align*}
(1 + r)^k p_r(\psi_\zeta(A)) &\leqslant (1 + r)^k p_r(\psi_\zeta(A - B)) + (1 + r)^k p_r(\psi_\zeta(B)) \\
&\leqslant 2M_k\ ||A - B||_k + (1 + r)^k p_r(\psi_\zeta(B)).
\end{align*}
Since $\bb{A}_{\rm loc}$ is dense in $\bb{A}_{(k)}$, the desired limit follows.	
	
\end{proof}

\medskip

\newpage

\appendix

\section{Basic facts on Frech\'et $\ast$-algebras}\label{app:Fre_Al}

A Fréchet algebra is a straightforward generalization of a Banach algebra. Recall that a Banach algebra $\bb{B}$ is an algebra over $\C$, with a norm $\|\cdot\|$ which satisfies $\| AB \| \leq \| A \| \| B \|$, for all $A, B \in B$, and which is complete for the topology given by the norm
 $\|\cdot\|$. In the following we will use the notation  $\N_0:=\{0\}\cup\N$.

\subsection{Frech\'et spaces}
To begin with, let $\bb{V}$ be a vector space over $\C$, and let $\{ \| \cdot \|_{k} \}_{k\in\N_0}$ be a sequence of seminorms on $\bb{V}$.
We say that $\bb{V}$ is a Fréchet space if $\bb{V}$ is complete for the topology given by the seminorms. This means that 
a sequence $\{A_{n}\}_{n\in\N_0}$ of elements of $\bb{V}$ 
converges to an element $A \in \bb{V}$ if for each $\epsilon >0$ and each $k \in {\mathbb N}_0$, there is some sufficiently large $N_{k,\epsilon} > 0$ such that $\| A_{n} - A \|_{k} \leq \epsilon$ for $n \geq N_{n, \epsilon}$. In other words $A_{n}$ tends to $A$ in each of the seminorms. This defines the topology on $\bb{V}$, and $\bb{V}$ is of course complete if and only if every Cauchy sequence in $\bb{V}$ converges to an element of $\bb{V}$.
A base of neighborhoods of zero for such a topology consists of sets
\[
\s{U}_{[(k_1,\epsilon_1)\ldots(k_p,\epsilon_p)]}\;:=\;\{A\in\bb{V}\;|\;\|A\|_{k_1}\leqslant\epsilon_1\;,\ldots\;, \|A\|_{k_p}\leqslant\epsilon_p\}\;.
\]
Recall that in a topological vector space, the topology is determined by a base of zero neighborhoods.

\medskip

Now let $\rr{N} = \{\| \cdot \|_{k} \}_{k\in\N_0}$ 
and $\rr{N}' = \{\| \cdot \|'_{k} \}_{k\in\N_0}$
be two sequences of seminorms on $\bb{V}$.  We say that
$\rr{N}$ and $\rr{N}'$ are equivalent sequences of seminorms if the topologies they induce on $\bb{V}$ are equivalent.
We say that $\rr{N}'$ is an increasing sequence of seminorms on $\bb{V}$ if $\| A \|'_{k} \leq \| A \|'_{k+1}$ for all $k \in \N_0$ and $A\in \bb{V}$.
\begin{proposition} If $\bb{V}$ is any vector space over $\C$ with seminorms $\rr{N} = \{\| \cdot \|_{k} \}_{k\in\N_0}$ , then there exists an equivalent increasing sequence  of seminorms  $\rr{N}' = \{\| \cdot \|'_{k} \}_{k\in\N_0}$  given by
\[
\| A \|_{k}^{\prime}\;: =\; \max_{p\leq k} \| A \|_{p}\;.
\]
\end{proposition}

\medskip

Once $\bb{V}$ has been endowed with an increasing sequence  of seminorms one can define a  
 base of neighborhoods of zero simply bay
\[
\s{U}_{(k,\epsilon)}\;:=\;\{A\in\bb{V}\;|\;\|A\|_{k}\leqslant\epsilon\}\;.
\]

\begin{proposition}
Two increasing sequences of seminorms 
$\rr{N} = \{\| \cdot \|_{k} \}_{k\in\N_0}$ and $\rr{N}' = \{\| \cdot \|'_{k} \}_{k\in\N_0}$
 on  $\bb{V}$ define equivalent topologies if and only if for every $k \in \N_0$, there exists $m\in \N_0$ and $C_{k}>0$ such that $\| A \|_{k} \leq C_{k} \| A \|_{m}^{\prime}$ and $\| A \|_{k}^{\prime} \leq C_{k} \| A \|_{m}$.
\end{proposition}

\subsection{Frech\'et $\ast$-algebras}
Next assume that $\bb{A}$ is a Frech\'et space 
with respect to a system of  increasing sequences of seminorms $\rr{N} = \{\| \cdot \|_{k} \}_{k\in\N_0}$ and 
endowed with a multiplication operation defined on it, so that $\bb{A}$ is a algebra over $\C$.  We say that $\bb{A}$ is a Fréchet algebra\footnote{Fréchet algebras, as  defined  here, are also called $B_0$-algebras.} if for each $k\in \N_0$, there exists a constant $C_{k}>0$ and a natural number $m\geq k$ such that 
\begin{equation} 
\| AB \|_{k}\; \leq\; C_{k} \| A \|_{m}\; \|B \|_{m}\;, \qquad A, B \in \bb{A}\;. 
\label{eq:mConvDef} \end{equation} 
This implies that the tultiplication is jointly continuous, \ie if 
$A_n\to A$ and $B_n\to B$ one has that $A_nB_n\to AB$.
Therefore, joint continuity of multiplication is part of the definition of a Fréchet algebra.
 A Fréchet algebra $\bb{A}$ may or may not have an identity (or unit) element ${\bf 1}$. If $\bb{A}$ is unital, we do not require that $\| {\bf 1} \|_{k} = 1$, as is often done for Banach algebras.

\medskip

A Fréchet $\ast$-algebra is a Fréchet algebra with a continuous involution.
\subsection{Rapidly decaying sequence}\label{app:Rap_seq}
Let $\s{S}(\N_0)$ be the space of complex valued sequences that decay at infinity faster than any power (Schwartz-type sequences), \ie
\[
\s{S}(\N_0)\;:=\;\{g:\N_0\to\C\;|\;|||g|||_k<\infty\;,\;\; \forall\; k\in\N_0 \}
\]
where the system of norms $|||\cdot|||_k$ is defined by
 \begin{equation}\label{eq:seq_nor}
|||g|||_k\;:=\;\sup_{r\in\N_0}\;(1+r)^k|g(r)|\;<\;+\infty\;,\qquad \forall\;k\in\N_0
\end{equation}
Since $(1+r)^{k+1}\geqslant (1+r)^k$
 it follows that $|||g|||_k\leqslant |||g|||_{k+1}$ for every $g\in\s{S}(\N_0)$, meaning that the system of norms is increasing.
\begin{proposition}\label{prop:compl}
 The vector space  $\s{S}(\N_0)$ turns out to be a  a Fr\'echet space when topologized by the family  of norms $|||\cdot|||_k$. 
\end{proposition}

\medskip

Observe that 
\[
(1+r)^k\;|f(r)g(r)|\;\leqslant\;(1+r)^{2k}\;|f(r)|\;|g(r)|
\]
one immediately deduces that 
\begin{equation}\label{eq:sub_nor}
|||fg|||_k\;\leqslant\;|||f|||_k\;|||g|||_k\;,\qquad \forall\; k\in\N_0\;
\end{equation}
\ie all the norms are submultiplicative with the pointwise product.
The space $\s{S}(\N_0)$ can also be endowed with the involution given by the complex conjugation. If $g\in\s{S}(\N_0)$ then $g^*$ is the sequence defined by $g^*(r):=\overline{g(r)}$. Clearly $g^*\in \s{S}(\N_0)$ and one can easily check that
\begin{equation}\label{eq:adg}
|||gg^*|||_k\;\leqslant\;|||g|||_k^2\;=\;|||g^*|||_k^2\;,\qquad \forall\; k\in\N_0\;.
\end{equation}
Furthermore, any function $f \in \s{S}(\N_0)$ can be bounded above by a monotonically decreasing function $g \in \s{S}(\N_0)$. For instance, one can define $g$ by:
\begin{equation}\label{eq: BoundAboveS}
 g(r) \; := \; \sup_{s \geqslant r} f(s).
\end{equation}
\begin{proposition}\label{pro:F-alg}
 The space  $\s{S}(\N_0)$ is a Fr\'echet $\ast$-algebra with respect to the pointwise multiplication and the complex conjugation.
\end{proposition}


\begin{thebibliography} {[RMCPV]}
\frenchspacing \baselineskip=12 pt plus 1pt minus 1pt


\bibitem[AM]{Araki-Fermion-Lattice-Systems}
    {Araki, H; Moriya, H}:
    {\sl Equilibrium statistical mechanics of fermion lattice systems}. In
    {\em Reviews in Mathematical Physics}. Vol {\bf 15} (02), 93-198 (2003).

\bibitem[AH]{araki1967collision}
{Araki, H. and Haag, R.}:
  {\sl Collision cross sections in terms of local observables}. In
{\em Communications in Mathematical Physics}. Vol {\bf 4}(2), 77--91 (1967).


\bibitem[BR1]{Bratteli-Robinson-1} {Bratteli, O.; Robinson, D. W.}: {\em  Operator Algebras and Quantum Statistical Mechanics 1. $C^*$- and $W^*$-Algebras, Symmetry Groups, Decomposition of States}. 
 Springer-Verlag, Berlin-Heidelberg, 1987.

\bibitem[BR2]{Bratteli-Robinson-2} {Bratteli, O.; Robinson, D. W.}: {\em  Operator Algebras and Quantum Statistical Mechanics 2. Equilibrium States. Models in Quantum Statistical Mechanics}. 
 Springer-Verlag, Berlin-Heidelberg, 1997.

	
\bibitem[BDN]{bachmann2016lieb}
{Bachmann, S.; Dybalski, W.; Naaijkens, P.}:
{\sl Lieb--Robinson bounds, Arveson spectrum and Haag--Ruelle scattering theory for gapped quantum spin systems.} In
{\em Annales Henri Poincar{\'e}}. Vol {\bf 17} (7), 1737--1791 (2016).

\bibitem[BdRF]{AdiabaticThm}
{Bachmann, Sven and De Roeck, Wojciech and Fraas, Martin}
{\sl The adiabatic theorem and linear response theory for extended quantum systems.} In
{\em Communications in Mathematical Physics}. Vol {\bf 361} (3), 997--1027 (2018).


	
\bibitem[EMNY]{Improved LR polynomial}
	{Else, D.; Machado, F.; Nayak, C.; Yao, N.}:
		{\sl Improved Lieb-Robinson bound for many-body Hamiltonians with power-law interactions.} In
		{\em Phys. Rev. A.} Vol {\bf 101}, 022333, (2020).
		
		
\bibitem[GNVW]{GrossNesmeVogtsWerner}
{Gross, D. and Nesme, V. and Vogts, H and Werner, R.F.}:
  {\sl Index theory of one dimensional quantum walks and cellular automata.} In
{\em Communications in Mathematical Physics} Vol {\bf 310} (2), 419--454 (2012).



\bibitem[HK]{HastingsPolLRB}
{Hastings, M. and Koma, T.}:
{\sl Spectral Gap and Exponential Decay of Correlations.} In
{\em Communications in Mathematical Physics} Vol {\bf 265}, 781--804 (2006).


\bibitem[KS1]{Kapustin-Sopenko-Lieb} 
    {Kapustin, A.;  Sopenko, N.}:
    {\sl Anomalous symmetries of quantum spin chains and a generalization of the Lieb-Schultz-Mattis theorem}.
    E-print: {\tt arXiv:2401.02533} (2024).
    

\bibitem[KS2]{Kapustin-Sopenko-Local-Noether}
    {Kapustin, A.;  Sopenko, N.}:
    {\sl Local Noether theorem for quantum lattice systems and topological invariants of gapped states}.
    In {\em J. Math. Phys.} {\bf{63}}, 091903 (2022).
    
\bibitem[Kub]{Kubota} 
    {Kubota, Y.}:
    {\sl Stable homotopy theory of invertible gapped quantum spin systems I: Kitaev's $\Omega$-spectrum}.
    E-print: {\tt arXiv:2503.12618} (2025).
    
\bibitem[KuwSa]{Kuwahara}
	{Kuwahara, T.; and Saito, K.}:
	{\sl Strictly Linear Light Cones in Long-Range Interacting Systems of Arbitrary Dimensions}.
	In {\em Phys. Rev. X} {\bf{10}}(3), 031010 (2020)
   
    

\bibitem[LR]{Lieb-Robinson Original}
	{Lieb, E.H.; Robinson, D.W.}: 
	{\sl The finite group velocity of quantum spin systems}. In {\em Commun. Math. Phys}. Vol {\bf 28}, 251-257 (1972).
	
	
	

\bibitem[Naa]{Naaijkens-peter}{Naaijkens, P}:
{\sl Quantum spin systems on infinite lattices}. Springer.


\bibitem[NSW]{nachtergaele-scholz-werner-13} 
{Nachtergaele, B.; Scholz, V. B.; Werner, R. F.}: 
{\sl Local Approximation of Observables and Commutator Bounds}. In
{\em  Operator Methods in Mathematical Physics} (eds. Janas, J.; Kurasov, P.; Laptev, A.; Naboko, S.). 
 Operator Theory: Advances and Applications, vol {\bf 227}. Birkh\"{a}user, Basel, 2013.


\bibitem[NS]{Nachtergaele-Sims Exponential Clustering}
	{Nachtergaele, B.; Sims, R.}: 
	{\sl Lieb-Robinson Bounds and the Exponential Clustering Theorem}. In
	{\em Commun. Math. Phys.} Vol {\bf 265}, 119–130, (2006).


\bibitem[NSY]{Nachtergaele-Sims-Young-Quasilocality-bounds}
    {Nachtergaele, B.; Sims, R.; Young, A.}:
    {\sl Quasi-locality bounds for quantum lattice systems. I. Lieb-Robinson bounds, quasi-local maps, and spectral flow automorphisms.} In
    {\em J. Math. Phys}. {\bf 60} (6), 061101, (2019).
  
    
\bibitem[NSO]{Nachtergaele-Sims-Ogata}
	{Nachtergaele, B.; Ogata, Y.; Sims, R.}:
	{\sl Propagation of Correlations in Quantum Lattice Systems}. In
	{\em  J Stat Phys}. Vol {\bf 124}, 1-13, (2006).
	

	

\bibitem[O]{ogata2021Cohom}
  {Ogata, Y.:}
  {\sl An-valued index of symmetry-protected topological phases with on-site finite group symmetry for two-dimensional quantum spin systems.} In
  {\em Forum of Mathematics, Pi}. Vol {\bf 9}, e13, (2021).



\bibitem[RaWi]{Raeburn Morita}
	{Raeburn, I.; Williams, D.:} {\em  Morita Equivalence and Continuous-Trace $C^*-$algebras}. 
Mathematical Surveys and Monographs, 1998.
    
    
\bibitem[RWW]{Ranarad-Walter-Witteveen-Converse-Lieb}
    {Ranard, D.; Walter, M.; Witteveen, F.}:
    {\sl A converse to Lieb-Robinson bounds in one dimension using index theory.}
    Ann. Henri Poincar\'{e} {\bf 23}, 3905–3979, (2022).

  
\bibitem[TGBEGAL]{LucasLongLRB}
{Tran, Minh C and Guo, Andrew Y and Baldwin, Christopher L and Ehrenberg, Adam and Gorshkov, Alexey V and Lucas, Andrew}:
{\sl Lieb-Robinson light cone for power-law interactions.} In
{\em Physical Review Letters}, {\bf 127}(16),{160401} (2021).

\bibitem[TW]{TeufelPolLRB}
{Tefuel, Stefan and Wessel, Tom}:
{\sl Lieb-Robinson bounds, automorphic equivalence and LPPL for long-range interacting fermions.}  E-print: {\tt arXiv:2507.03319} (2025).


 \end{thebibliography}
\end{document}